%% file: paper.tex
\documentclass[11pt]{article}
\usepackage{amsmath,amssymb,amsthm}
\usepackage{enumitem}
\usepackage{graphicx}
\usepackage{float}
\usepackage{fullpage}
\usepackage{xspace}
\usepackage{url}
\usepackage{nicefrac}
\usepackage{algorithm}
\usepackage{algpseudocode}
\usepackage[usenames]{color}

\newenvironment{tight-itemize}
{\begin{list}{$\bullet$}
{\itemsep=1pt plus 1pt minus 1pt\parsep=0pt\labelsep=4pt
\parsep=0pt\def\makelabel##1{\hss\llap##1}}}
{\end{list}}

\DeclareMathOperator*{\E}{\mathbb{E}}
\let\Pr\relax
\DeclareMathOperator*{\Pr}{\mathbb{P}}

\DeclareMathOperator{\poly}{poly}

\newcommand{\widebar}{\overline}
\newcommand{\cluster}{W}
\newcommand{\lap}{\mathcal{L}}
\newcommand{\vol}{\mathop{vol}}
\newcommand{\eps}{\varepsilon}

\newcommand{\ceil}[1]{\lceil #1 \rceil}

\newcommand{\oracle}{\mathcal{O}}
\newcommand{\R}{\mathbb{R}}

\newcommand{\projeps}{{\proj{\nicefrac[]{1}{\eps^2}}}}
\newcommand{\projepstwo}[1]{{\proj{\nicefrac[]{1}{#1^2}}}}

\newcommand{\cm}{\textsc{CountMin}\xspace}
\newcommand{\cs}{\textsc{CountSketch}\xspace}
\newcommand{\hcs}{\textsc{Hierarchical CountSketch}\xspace}
\newcommand{\pcs}{\textsc{PartitionCountSketch}\xspace}
\newcommand{\es}{\textsc{ExpanderSketch}\xspace}

\newcommand{\enc}{\mathrm{enc}}

\newcommand{\proj}[1]{[\overline{#1}]}
\renewcommand{\log}{\lg}

\newcommand{\cutandclose}{\textsc{CutGrabClose}\xspace}

\newtheorem{theorem}{Theorem}
\newtheorem{remark}{Remark}
\newtheorem{lemma}{Lemma}
\newtheorem{corollary}{Corollary}
\newtheorem{definition}{Definition}

\newcommand{\proofbelow}{3pt}
\newcommand{\afterproof}{\hfill $\blacksquare$ \par \vspace{\proofbelow}}
\renewenvironment{proof}{\noindent\textbf{Proof.}\,}{\afterproof}

\newcommand{\EquationName}[1]{\label{eq:#1}}
\newcommand{\DefinitionName}[1]{\label{def:#1}}

\newcommand{\LemmaName}[1]{\label{lem:#1}}
\newcommand{\CorollaryName}[1]{\label{cor:#1}}
\newcommand{\SectionName}[1]{\label{sec:#1}}
\newcommand{\TheoremName}[1]{\label{thm:#1}}
\newcommand{\RemarkName}[1]{\label{rem:#1}}
\newcommand{\FigureName}[1]{\label{fig:#1}}

\newcommand{\Equation}[1]{Eq.\:\eqref{eq:#1}}

\newcommand{\Definition}[1]{Definition~\ref{def:#1}}
\newcommand{\Lemma}[1]{Lemma~\ref{lem:#1}}
\newcommand{\Corollary}[1]{Corollary~\ref{cor:#1}}
\newcommand{\Section}[1]{Section~\ref{sec:#1}}
\newcommand{\Theorem}[1]{Theorem~\ref{thm:#1}}
\newcommand{\Remark}[1]{Remark~\ref{rem:#1}}

\newcommand{\Figure}[1]{Figure~\ref{fig:#1}}

\newcommand{\Eqsub}[1]{\eqref{eq:#1}}

 \author{Kasper Green Larsen\thanks{Aarhus University. \texttt{larsen@cs.au.dk}. Supported by Center for Massive Data Algorithmics, a Center of the Danish National Research Foundation, grant DNRF84, a Villum Young Investigator Grant and an AUFF Starting Grant.}
   \and Jelani Nelson\thanks{Harvard University. \texttt{minilek@seas.harvard.edu}. Supported by NSF grant IIS-1447471 and
   CAREER award CCF-1350670, ONR Young Investigator award N00014-15-1-2388, and a Google Faculty Research Award.}
 \and Huy L. Nguy$\tilde{\hat{\mbox{e}}}$n\thanks{Toyota Technological Institute at Chicago. \texttt{hlnguyen@cs.princeton.edu}.}
 \and Mikkel Thorup\thanks{University of
   Copenhagen. \texttt{mikkel2thorup@gmail.com}. Supported in part by
   Advanced Grant DFF-0602-02499B from the Danish Council for
   Independent Research under the Sapere Aude research career programme.}.}

\title{Heavy hitters via cluster-preserving clustering}

\begin{document}

\setcounter{page}{0}

\maketitle

\thispagestyle{empty}

\begin{abstract}
In the turnstile $\ell_p$ heavy hitters problem with parameter $\eps$, one must maintain a high-dimensional vector $x\in\R^n$ subject to updates of the form $\texttt{update}(i,\Delta)$ causing the change $x_i\leftarrow x_i + \Delta$, where $i\in[n]$, $\Delta\in\R$. Upon receiving a query, the goal is to report every ``heavy hitter'' $i\in[n]$ with $|x_i| \ge \eps \|x\|_p$ as part of a list $L\subseteq[n]$ of size $O(1/\eps^p)$, i.e.\ proportional to the maximum possible number of heavy hitters. In fact we solve the stronger tail version, in which $L$ should include every $i$ such that $|x_i| \ge \eps \|x_{\proj{\nicefrac[]{1}{\eps^p}}}\|_p$, where $x_{\proj{k}}$ denotes the vector obtained by zeroing out the largest $k$ entries of $x$ in magnitude.

For any $p\in(0,2]$ the \cs of \cite{CharikarCF04} solves $\ell_p$ heavy hitters using $O(\eps^{-p}\log n)$ words of space with $O(\log n)$ update time, $O(n\log n)$ query time to output $L$, and whose output after any query is correct with high probability (whp) $1 - 1/\poly(n)$ \cite[Section 4.4]{JowhariST11}. This space bound is optimal even in the strict turnstile model \cite{JowhariST11} in which it is promised that $x_i\ge 0$ for all $i\in[n]$ at all points in the stream, but unfortunately the query time is very slow. To remedy this, the work \cite{CormodeM05} proposed the ``dyadic trick'' for the \cm sketch for $p=1$ in the strict turnstile model, which to maintain whp correctness achieves suboptimal space $O(\eps^{-1}\log^2 n)$, worse update time $O(\log^2 n)$, but much better query time $O(\eps^{-1}\poly(\log n))$. An extension to all $p\in(0,2]$ appears in \cite[Theorem 1]{KaneNPW11}, and can be obtained from \cite{Pagh13}.

We show that this tradeoff between space and update time versus query time is unnecessary. We provide a new algorithm, \es, which in the most general turnstile model achieves optimal $O(\eps^{-p}\log n)$ space, $O(\log n)$ update time, and fast $O(\eps^{-p}\poly(\log n))$ query time, providing correctness whp. In fact, a simpler version of our algorithm for $p=1$ in the strict turnstile model answers queries even {\em faster} than the ``dyadic trick'' by roughly a $\log n$ factor, dominating it in all regards. Our main innovation is an efficient reduction from the heavy hitters to a clustering problem in which each heavy hitter is encoded as some form of noisy spectral cluster in a much bigger graph, and the goal is to identify every cluster. Since every heavy hitter must be found, correctness requires that every cluster be found. We thus need a ``cluster-preserving clustering'' algorithm, that partitions the graph into clusters with the promise of not destroying any original cluster. To do this we first apply standard spectral graph partitioning, and then we use some novel combinatorial techniques to modify the cuts obtained so as to make sure that the original clusters are sufficiently preserved.  Our cluster-preserving clustering may be of broader interest much beyond heavy hitters.
\end{abstract}

\newpage

\input{intro.tex}
\input{previous.tex}
\input{prelim.tex}
\input{overview.tex}
\input{turnstile.tex}
\input{cluster.tex}

\section*{Acknowledgments}
We thank Noga Alon for pointing us to \cite[Lemma 2.3]{AlonC06}, Jonathan Kelner for the reference \cite{OrecchiaSV12}, Lorenzo Orecchia and Sushant Sachdeva for answering several questions about \cite{OrecchiaV11,OrecchiaSV12}, Piotr Indyk for the reference \cite{GilbertLPS14}, Graham Cormode for the reference \cite{Pagh13}, Yi Li for answering several questions about \cite{GilbertLPS14}, Mary Wootters for making us aware of the formulation of the list-recovery problem in coding theory and its appearance in prior work in compressed sensing and group testing, Atri Rudra for useful discussions on list-recovery algorithms for Parvaresh-Vardy and Folded Reed-Solomon codes, and Fan Chung Graham and Olivia Simpson for useful conversations about graph partitioning algorithms.

\bibliographystyle{alpha}
\bibliography{biblio}

\appendix

\input{appendix.tex}

\end{document}

%% file: intro.tex
\section{Introduction}\SectionName{intro}
Finding {\em heavy hitters} in a data stream, also known as {\em elephants} or {\em frequent items}, is one of the most practically important and core problems in the study of streaming algorithms. The most basic form of the problem is simple: given a long stream of elements coming from some large universe, the goal is to report all frequent items in the stream using an algorithm with very low memory consumption, much smaller than both the stream length and universe size.

In practice, heavy hitters algorithms have been used to find popular destination addresses and heavy bandwidth users by AT{\&}T \cite{JohnsonCKMSS04}, to find popular search query terms at Google \cite{PikeDGQ05}, and to answer so-called ``iceberg queries'' in databases \cite{FangSGMU98} (to name a few of many applications). A related problem is finding superspreaders~\cite{LiuCG13}, the ``heavy hitters'' in terms of number of distinct connections, not total bandwidth.

In the theoretical study of streaming algorithms, finding heavy hitters has played a key role as a subroutine in solving many other problems. For example, the first algorithm providing near-optimal space complexity for estimating $\ell_p$ norms in a stream for $p > 2$ needed to find heavy hitters as a subroutine \cite{IndykW05}, as did several follow-up works on the same problem \cite{GangulyB09,AndoniKO11,BravermanO13,BravermanKSV14,Ganguly15}. Heavy hitters algorithms are also used as a subroutine in streaming entropy estimation \cite{ChakrabartiCM10,HarveyNO08}, $\ell_p$-sampling \cite{MonemizadehW10}, cascaded norm estimation and finding block heavy hitters \cite{JayramW09}, finding duplicates \cite{GopalanR09,JowhariST11}, fast $\ell_p$ estimation for $0 < p < 2$ \cite{NelsonW10,KaneNPW11}, and for estimating more general classes of functions of vectors updated in a data stream \cite{BravermanO10a,BravermanC15}.

In this work we develop a new heavy hitters algorithm, \es, which reduces the heavy hitters problem to a new problem we define concerning graph clustering, which is an area of interest in its own right. Our special goal is to identify {\em all} clusters of a particular type, even if they are much smaller than the graph we are asked to find them in. While our algorithm as presented is rather theoretical with large constants, the ideas in it are simple with potential practical impact.

The particular formulation we focus on is finding clusters with low external conductance in the graph, and with good connectivity properties internally (a precise definition of these clusters, which we call {\em $\epsilon$-spectral clusters}, will be given soon). Similar models have been used for finding a community within a large network, and algorithms for identifying them have applications in community detection and network analysis. In computer vision, many segmentation algorithms such as~\cite{shi2000normalized,arbelaez2011contour} model images as graphs and use graph partitioning to segment the images. In theory, there has been extensive study on the problem with many approximation algorithms~\cite{leighton1988approximate,arora2009expander}. Our focus in this work is on theoretical results, with constants that may be unacceptable in practice.

%% file: previous.tex
\section{Previous work}\SectionName{previous}
In this work we consider the $\ell_p$ heavy hitters problem for $0 < p\le 2$. A vector $x\in\R^n$ is maintained, initialized to the all zero vector. A parameter $\eps\in(0, 1/2)$ is also given. This is a data structural problem with updates and one allowed query that are defined as follows.
\begin{itemize}
\item \texttt{update($i, \Delta$):} Update$x_i\leftarrow x_i + \Delta$, where $i\in[n]$ and $\Delta\in\R$ has finite precision.
\item \texttt{query():} Return a set $L\subseteq[n]$ of size $|L| = O(\eps^{-p})$ containing all $\eps$-heavy hitters $i\in[n]$ under $\ell_p$. Here we say $i$ is an $\eps$-heavy hitter under $\ell_p$ if $|x_i| \ge \eps \|x_{\proj{1/\eps^p}}\|_p$, where $x_{\proj{k}}$ denotes the vector $x$ with the largest $k$ entries (in absolute value) set to zero. Note the number of heavy hitters never exceeds $2/\eps^p$.
\end{itemize}

Unless stated otherwise we consider randomized data structures in which any individual query has some failure probability $\delta$. Note our definition of $\eps$-heavy hitters includes all coordinates satisfying the usual definition requiring $|x_i| \ge \eps \|x\|_p$, and potentially more. One can thus recover all the usual heavy hitters from our result by, in post-processing, filtering $L$ to only contain indices that are deemed large by a separate \cs run in parallel. A nice feature of the more stringent version we solve is that, by a reduction in \cite[Section 4.4]{JowhariST11}, solving the $\ell_2$ version of the problem with error $\eps' = \eps^{p/2}$ implies a solution for the $\ell_p$ $\eps$-heavy hitters for any $p\in (0, 2]$. Hence for the remainder of the assume we discuss $p=2$ unless stated otherwise. We also note that any solution to $\ell_p$ heavy hitters for $p>2$, even for constant $\eps, \delta$, must use polynomial space $\Omega(n^{1-2/p})$ bits \cite{BJKS04}.

Before describing previous work, we first describe, in the terminology of \cite{Muthukrishnan05}, three different streaming models that are frequently considered.

\begin{itemize}
\item \textbf{Cash-register:} This model is also known as the {\em insertion-only} model, and it is characterized by the fact that all updates \texttt{update}$(i,\Delta)$ have $\Delta = 1$.
\item \textbf{Strict turnstile:} Each update $\Delta$ may be an arbitrary positive or negative number, but we are promised that $x_i\ge 0$ for all $i\in[n]$ at all points in the stream.
\item \textbf{General turnstile:} Each update $\Delta$ may be an arbitrary positive or negative number, and there is no promise that $x_i\ge 0$ always. Entries in $x$ may be negative.
\end{itemize}

It should be clear from the definitions that algorithms that work correctly in the general turnstile model are the most general, followed by strict turnstile, then followed by the cash-register model. Similarly, lower bounds proven in the cash-register model are the strongest, and those proven in the general turnstile model are the weakest. The strict turnstile model makes sense in situations when deletions are allowed, but items are never deleted more than they are inserted. The general turnstile model is useful when computing distances or similarity measures between two vectors $z, z'$, e.g.\ treating updates $(i,\Delta)$ to $z$ as $+\Delta$ and to $z'$ as $-\Delta$, so that $x$ represents $z - z'$.

\begin{figure}
\begin{center}
\begin{tabular}{|c|c|c|c|c|c|}
\hline
reference & space & update time & query time & randomized? & norm\\
\hline
\cite{CormodeM05}\textcolor{white}{$^*$} & $\eps^{-1}\lg n$ & $\log n$ & $n\log n$ & Y & $\ell_1$\\
\cline{2-5}
\cite{CormodeM05}$^*$ & $\eps^{-1}\lg^2 n$ & $\lg^2 n$ & $\eps^{-1}\lg^2 n$& Y & $\ell_1$\\
\hline
\cite{NNW14} & $\eps^{-2}\lg n$ & $\eps^{-1}\lg n$ & $n\lg n$ & N & $\ell_1$\\
\hline
\cite{CharikarCF04} & $\eps^{-p}\lg n$ & $\log n$ & $n\log n$ & Y & $\ell_p$, $p\in(0,2]$\\
\hline
\cite{KaneNPW11,Pagh13}$^*$ & $\eps^{-p}\lg^2 n$ & $\log^2 n$ & $\eps^{-p}\log^2 n$ & Y & $\ell_p$, $p\in(0,2]$\\
\hline
\cite{CormodeH08} & $\eps^{-p}\lg n$ & $\log n$ & $\eps^{-p} \cdot n^{\gamma}$ & Y & $\ell_p$, $p\in(0,2]$\\
\hline
\textbf{This work} & $\eps^{-p}\log n$ & $\lg n$ & $\eps^{-p}\poly(\log n)$ & Y & $\ell_p$, $p\in(0,2]$\\
\hline
\end{tabular}
\caption{Comparison of turnstile results from previous work and this work for $\ell_p$ heavy hitters when desired failure probability is $1/\mathrm{poly}(n)$. Space is measured in machine words. For \cite{CormodeH08}, $\gamma>0$ can be an arbitrarily small constant. If a larger failure probability $\delta \gg 1/\poly(n)$ is tolerable, one $\log n$ factor in space, update time, and query time in each row with an asterisk can be replaced with $\log((\log n)/(\eps\delta))$. For \cite{Pagh13}, one $\log n$ in each of those bounds can be replaced with $\log(1/(\eps\delta))$. For \cite{CormodeH08}, the query time can be made $(\eps^{-p}\log n)((\log n)/(\eps\delta))^\gamma$.}\FigureName{prev-table}
\end{center}
\end{figure}

For the heavy hitters problem, the first algorithm with provable guarantees was a deterministic algorithm for finding $\ell_1$ heavy hitters in the cash-register model \cite{MisraG82}, for solving the non-tail version of the problem in which heavy hitters satisfy $|x_i| \ge \eps\|x\|_1$. The space complexity and query time of this algorithm are both $O(1/\eps)$, and the algorithm can be implemented so that the update time is that of insertion into a dynamic dictionary on $O(1/\eps)$ items \cite{DemaineLM02} (and thus expected $O(1)$ time if one allows randomness, via hashing). We measure running time in the word RAM model and space in machine words, where a single word is assumed large enough to hold the maximum $\|x\|_\infty$ over all time as well as any $\log n$ bit integer. The bounds achieved by \cite{MisraG82} are all optimal, and hence $\ell_1$ heavy hitters in the cash-register model is completely resolved. An improved analysis of \cite{BerindeICS10} shows that this algorithm also solves the tail version of heavy hitters. In another version of the problem in which one simultaneously wants the list $L$ of $\eps$ heavy hitters together with additive $\eps'\|x\|_1$ estimates for every $i\in L$, \cite{BhattacharyyaDW16} gives space upper and lower bounds which are simultaneously optimal in terms of $\eps$, $\eps'$, $n$, and the stream length.

For $\ell_2$ heavy hitters in the cash-register model, the \cs achieves $O(\eps^{-2}\lg n)$ space. Recent works gave new algorithms using space $O(\eps^{-2}\log(1/\eps)\lg\lg n)$ \cite{BravermanCIW16} and $O(\eps^{-2}\log(1/\eps))$ \cite{braverman2016bptree}. A folklore lower bound in the cash-register model is $\Omega(1/\eps^2)$ machine words, which follows since simply encoding the names of the up to $1/\eps^2$ heavy hitters requires $\Omega(\lg(\binom{n}{1/\eps^2}))$ bits.

Despite the $\ell_1$ and $\ell_2$ heavy hitters algorithms above in the cash-register model requiring $o(\eps^{-p}\log n)$ machine words, it is known finding $\eps$-heavy hitters under $\ell_p$ for any $0<p\le 2$ in the strict turnstile model requires $\Omega(\eps^{-p}\log n)$ words of memory, even for constant probability of success \cite{JowhariST11}. This optimal space complexity is achieved even in the more general turnstile model by the \cs of \cite{CharikarCF04} for all $p\in(0,2]$ (the original work \cite{CharikarCF04} only analyzed the \cs for $p=2$, but a a very short argument of \cite[Section 4.4]{JowhariST11} shows that finding $\ell_p$ heavy hitters for $p\in(0,2)$ reduces to finding $\ell_2$ heavy hitters by appropriately altering $\eps$).

For any $p\in(0,2]$, the \cs achieves optimal $O(\eps^{-p}\lg n)$ space, the update time is $O(\lg n)$, the success probability of any query is with high probability (whp) $1 - 1/\poly(n)$, but the query time is a very slow $\Theta(n\lg n)$. In \cite{CormodeM05}, for $p=1$ in the strict turnstile model the authors described a modification of the \cm sketch they dubbed the ``dyadic trick'' which maintains whp correctness for queries and significantly improves the query time to $O(\eps^{-1}\lg^2 n)$, but at the cost of worsening the update time to $O(\lg^2 n)$ and space to a suboptimal $O(\eps^{-1}\lg^2 n)$.  An easy modification of the dyadic trick extends the same bounds to the general turnstile model and also to $\eps$-heavy hitters under $\ell_p$ for any $0<p\le 2$ with the same bounds \cite[Theorem 1]{KaneNPW11} (but with $\eps^{-1}$ replaced by $\eps^{-p}$). A different scheme using error-correcting codes in \cite[Theorem 4.1]{Pagh13} achieves the same exact bounds for $\ell_p$ heavy hitters when whp success is desired. This tradeoff in sacrificing space and update time for better query time has been the best known for over a decade for any $0<p\le 2$.

From the perspective of {\em time} lower bounds, \cite{LarsenNN15} showed that any ``non-adaptive'' turnstile algorithm for constant $\eps$ and using $\poly(\log n)$ space must have update time $\Omega(\sqrt{\log n/\log\log n})$. A non-adaptive algorithm is one in which, once the randomness used by the algorithm is fixed (e.g.\ to specify hash functions), the cells of memory probed when processing \texttt{update}($i,\Delta$) depend only on $i$ and not on the history of the algorithm. Note every known turnstile algorithm for a non-promise problem is non-adaptive.

In summary, there are several axes on which to measure the quality of a heavy hitters algorithm: space, update time, query time, and failure probability. The ideal algorithm should achieve optimal space, fast update time ($O(\log n)$ is the best we know), nearly linear query time e.g.\ $O(\eps^{-p}\poly(\log n))$, and $1/\poly(n)$ failure probability. Various previous works were able to achieve various subsets of at most three out of four of these desiderata, but none could achieve all four simultaneously.

\paragraph{Our contribution I:} We show the tradeoffs between space, update time, query time, and failure probability in previous works are unnecessary (see \Figure{prev-table}). Specifically, in the most general turnstile model we provide a new streaming algorithm, \es, which for any $0<p\le 2$ provides whp correctness for queries with $O(\lg n)$ update time, $O(\eps^{-p}\poly(\lg n))$ query time, and optimal $O(\eps^{-p}\log n)$ space. In fact in the strict turnstile model and for $p=1$, we are able to provide a simpler variant of our algorithm with query time $O(\eps^{-1}\log^{1+\gamma} n)$ for any constant $\gamma>0$, answering queries even {\em faster} than the fastest previous known algorithms achieving suboptimal update time and space, by nearly a $\log n$ factor, thus maintaining whp correctness while dominating it in all of the three axes: space, update time, and query time (see \Section{strict}).

\subsection{Cluster-preserving clustering}
Our algorithm \es operates by reducing the heavy hitters problem to a new clustering problem we formulate of finding clusters in a graph, and we then devise a new algorithm \cutandclose which solves that problem. Specifically, the \es first outputs a graph in which each heavy hitter is encoded by a well-connected cluster (which may be much smaller than the size of the total graph), then \cutandclose recovers the clusters, i.e.\ the heavy hitters, from the graph. There have been many works on finding clusters $S$ in a graph such that the conductance of the cut $(S, \bar{S})$ is small. We will only mention some of them with features similar to our setting. Roughly speaking, our goal is to find {\em all} clusters that are low conductance sets in the graph, and furthermore induce subgraphs that satisfy something weaker than having good spectral expansion. It is necessary for us to (1) be able to identify {\em all} clusters in the graph; and (2) have a quality guarantee that does not degrade with the number of clusters nor the relative size of the clusters compared with the size of the whole graph. As a bonus, our algorithm is also (3) able to work without knowing the number of clusters.

In the context of heavy hitters, requirement (1) arises since all heavy hitters need to be returned. The number of heavy hitters is not known and can be large, as is the ratio between the size of the graph and the size of a cluster (both can be roughly the square root of the size of our graph), leading to the requirement (2) above. One line of previous works~\cite{spielman2004nearly,andersen2009finding,gharan2012approximating} gives excellent algorithms for finding small clusters in a much larger graph with runtime proportional to the size of the clusters, but their approximation guarantees depend on the size of the whole graph, violating our requirement (2). In fact, requirement (2) is related to finding non-expanding small sets in a graph, an issue at the forefront of hardness of approximation~\cite{raghavendra2010graph}. As resolving this issue in general is beyond current techniques, many other works including ours attempt to solve the problem in structured special cases.  There are many other excellent works for graph clustering, but whose performance guarantees deterioriate as $k$ increases, also violating (2). The work~\cite{gharan2014partitioning} shows results in the same spirit as ours: if there are $k$ clusters and a multiplicative $\poly(k)$ gap between the $k$th and $(k+1)$st smallest eigenvalues of the Laplacian of a graph, then it is possible to partition the graph into at most $k$ high conductance induced subgraphs. Unfortunately, such a guarantee is unacceptable for our application due to more stringent requirements on the clusters being required as $k$ increases, i.e.\ the $\poly(k)$ gap in eigenvalues --- our application provides no such promise. Another work with a goal similar in spirit to ours is \cite{PengSZ15}, which given the conductances in our graphs deriving from heavy hitters would promise to find $k$ clusters $W$ up to error $\poly(k)\cdot |W|$ symmetric difference each. Unfortunately such a result is also not applicable to our problem, since in our application we could have $k \gg |W|$ so that the guarantee becomes meaningless.

A different line of works~\cite{MakarychevMV12, MakarychevMV14} give algorithms with constant approximation guarantees but without being able to identify the planted clusters, which is not possible in their models in general. In their setting, edges inside each cluster are adversarially chosen but there is randomness in the edges between clusters. Our setting is the opposite: the edges inside each cluster have nice structure while the edges between clusters are adversarial as long as there are few of them.

From the practical point of view, a drawback with several previous approaches is the required knowledge of the number of clusters. While perhaps not the most pressing issue in theory, it is known to cause problems in practice. For instance, in computer vision, if the number of clusters is not chosen correctly, algorithms like $k$-means tend to break up uniform regions in image segmentation~\cite{arbelaez2011contour}. Instead, a preferred approach is hierarchical clustering~\cite{arbelaez2011contour}, which is in the same spirit as our solution.

\begin{definition}\DefinitionName{spec-cluster}
An {\em $\epsilon$-spectral cluster} is a vertex set $W\subseteq V$ of any size satisfying the following two conditions:
First, only an $\epsilon$-fraction of the edges incident to $W$ leave $W$, that is, $|\partial(W)|\leq \epsilon \vol(W)$, where $\vol(W)$ is the sum of edge degrees of vertices inside $W$.
Second, given any subset $A$ of $W$, let 
$r=\vol(A)/\vol(W)$ and $B=W\setminus A$. Then
\[|E(A,B)|\geq(r(1-r)-\epsilon)\vol(W).\]
Note $r(1-r)\vol(W)$ is the number of edges one would expect to see between $A$ and $B$ had $W$ been a random graph with a prescribed degree distribution.
\end{definition}

\paragraph{Our contribution II:} We give a hierarchical graph partitioning algorithm combining traditional spectral methods and novel local search moves that can identify {\em all} $\epsilon$-spectral clusters in a potentially much larger graph for $\epsilon$ below some universal constant, independent of the number of clusters or the ratio between the graph size and cluster sizes. More precisely, consider a graph $G=(V,E)$.  In polynomial time, our algorithm can partition the vertices in a graph $G=(V,E)$ into subsets $V_1,\ldots V_\ell$ so that {\em every} $\epsilon$-spectral cluster $W$ matches some subset $V_i$
up to an $O(\epsilon\vol(W))$ symmetric difference. The algorithm is described in \Section{cluster}, with guarantee given by \Theorem{cluster} below.

\bigskip

\begin{theorem}\label{thm:cluster} For any given $\epsilon\leq 1/2000000$ and graph $G=(V,E)$, in polynomial time and linear space, we can find a family of 
disjoint vertex sets $U_1,\ldots U_\ell$ so that every
{\em $\epsilon$-spectral cluster\/} $\cluster$ of $G$ matches some set $U_i$ in the
sense that
\begin{itemize}
\item $\vol(\cluster\setminus U_i)\leq 3\,\epsilon\vol(\cluster)$.
\item $\vol(U_i\setminus \cluster)\leq 2250000\,\epsilon\vol(\cluster)$
\end{itemize}
Moreover, if $v\in U_i$, then $>4/9$ of its neighbors are in $U_i$.
\end{theorem}

An interesting feature of our solution is that in optimization in practice, it is usually expected to perform local search to finish off any optimization algorithm but it is rarely the case that one can prove the benefit of the local search step. In contrast, for our algorithm, it is exactly our local search moves which guarantee that every cluster is found. 

%% file: prelim.tex
\section{Preliminaries}\SectionName{prelim}
The letter $C$ denotes a positive constant that may change from line to line, and $[n] = \{1,\ldots,n\}$. We use ``with high probability'' (whp) to denote probability $1 - 1 /\poly(n)$.

All graphs $G = (V, E)$ discussed are simple and undirected. We say that a graph is a {\em $\lambda$-spectral expander} if the second largest eigenvalue, in magnitude, of its unnormalized adjacency matrix is at most $\lambda$. Given two vertex sets $S,T\subseteq V$, we define $E(S,T)=E\cap (S\times T)$. For $S$, we define the {\em volume} $\vol(S) = \sum_{v\in S} \deg(v)$, and {\em boundary} $\partial S=E(S,\widebar S)$ where $\widebar{S}=V\setminus S$. We specify $\vol_G$ and $\partial_G$ when the graph may not be clear from context. We define the {\em conductance} of the cut $(S, \widebar{S})$ as
$$
\phi(G, S) =  \frac{|\partial S|}{\min\{\vol(S), \vol(\widebar{S})\}}. 
$$
If $S$ or $\widebar S$ is empty, the conductance is $\infty$. We define the conductance of the entire graph $G$ by
$$
\phi(G) = \min_{\emptyset\subsetneq S\subsetneq  V} \phi(G, S) .
$$

%% file: overview.tex
\section{Overview of approach}\SectionName{overview}

We now provide an overview of our algorithm, \es. In \Section{reduction} we show turnstile heavy hitters with arbitrary $\eps$ reduces to several heavy hitters problems for which $\eps$ is ``large''. Specifically, after the reduction whp we are guaranteed each of the $\eps$-heavy hitters $i$ in the original stream now appears in some substream updating a vector $x'\in\R^n$, and $|x'_i| \ge (1/\sqrt{C\log n})\|x'_{\proj{C\log n}}\|_2$. That is, $i$ is an $\Omega(1/\sqrt{\log n})$-heavy hitter in $x'$. One then finds all the $\Omega(1/\sqrt{\log n})$ heavy hitters in each substream then outputs their union as the final query result. For the remainder of this overview we thus focus on solving $\eps$-heavy hitters for $\eps > 1/ \sqrt{C \log n}$, so there are at most $2/\eps^2 = O(\log n)$ heavy hitters. Our goal is to achieve failure probability $1/\poly(n)$ with space $O(\eps^{-2}\log n)$, update time $O(\log n)$, and query time $\eps^{-2}\poly(\log n) = \poly(\log n)$.

There is a solution, which we call the ``$b$-tree'' (based on \cite{CormodeH08}), which is almost ideal (see \Section{btree} and \Corollary{btree}). It achieves $O(\eps^{-2}\log n)$ space and $O(\log n)$ update time, but query time inverse polynomial in the failure probability $\delta$, which for us is $\delta = 1/\poly(n)$. Our main idea is to reduce to the case of much larger failure probability (specifically $1/\poly(\log n)$) so that using the $b$-tree then becomes viable for fast query. We accomplish this reduction while maintaining $O(\log n)$ update time and optimal space by reducing our heavy hitter problem into $m = \Theta(\log n/\log\log n)$ separate {\em partition heavy hitters problems}, which we now describe. 

In the partition $\eps$-heavy hitters problem, there is some parameter $\eps\in(0,1/2)$. There is also some partition $\mathcal{P} = \{S_1,\ldots,S_N\}$ of $[n]$, and it is presented in the form of an oracle $\oracle:[n]\rightarrow[N]$ such that for any $i\in[n]$, $\oracle(i)$ gives the $j\in[N]$ such that $i\in S_j$. In what follows the partitions will be random, with $\oracle_j$ depending on some random hash functions (in the case of adapting the $b$-tree to partition heavy hitters the $\oracle_j$ are deterministic; see \Section{btree}). Define a vector $y\in\R^N$ such that for each $j\in[N]$, $y_j = \|x_{S_j}\|_2$, where $x_S$ is the projection of $x$ onto a subset $S$ of coordinates. The goal then is to solve the $\eps$-heavy hitters problem on $y$ subject to streaming updates to $x$: we should output a list $L\subset[N]$, $|L| = O(1/\eps^2)$, containing all the $\eps$-heavy hitters of $y$. See \Section{partition} for more details.

Now we explain how we make use of partition heavy hitters. For each index $i\in[n]$ we can view $i$ as a length $\log n$ bitstring. Let $\enc(i)$ be an encoding of $i$ into $T = O(\log n)$ bits by an error-correcting code with constant rate that can correct an $\Omega(1)$-fraction of errors. Such codes exist with linear time encoding and decoding \cite{Spielman96}. We partition $\enc(i)$ into $m$ contiguous bitstrings each of length $t = T/m = \Theta(\log\log n)$. We let $\enc(i)_j$ for $j\in[m]$ denote the $j$th bitstring of length $t$ when partitioning $\enc(i)$ in this way. For our data structure, we instantiate $m$ separate partition heavy hitter data structures from \Section{btree} (a partition heavy hitter variant of the $b$-tree), $P_1,\ldots,P_m$, each with failure probability $1/\poly(\log n)$. What remains is to describe the $m$ partitions $\mathcal{P}_j$. We now describe these partitions, where each will have $N = 2^{O(t)} = \poly(\log n)$ sets.

\begin{figure}
\begin{center}
\scalebox{0.55}{\includegraphics{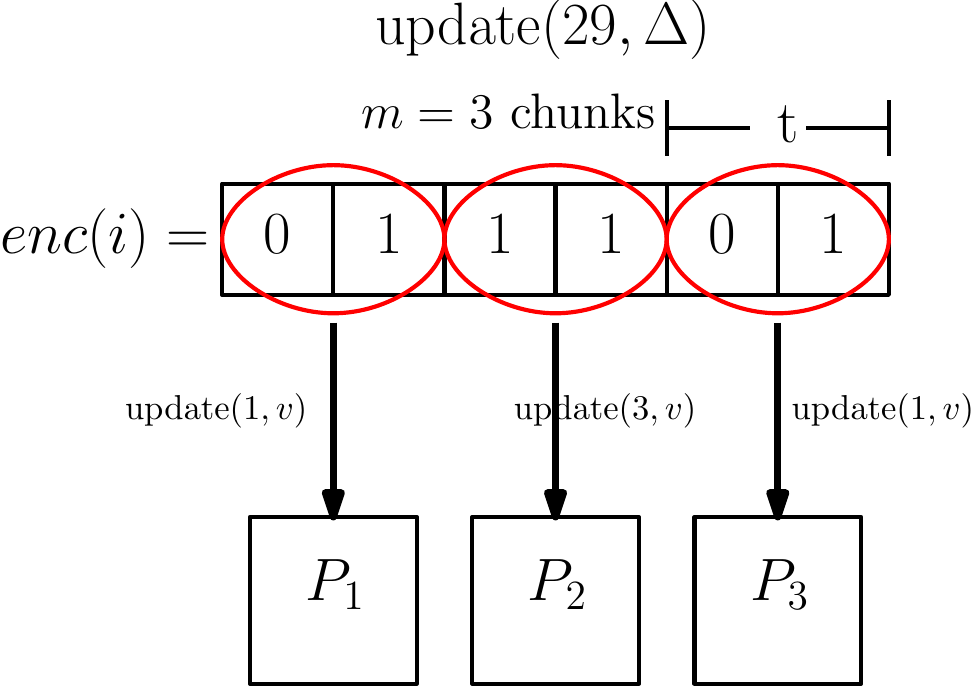}}
\end{center}
\caption{Simplified version of final data structure. The update is $x_{29} \leftarrow x_{29} + \Delta$ with $m = 3$, $t = 2$ in this example. Each $P_j$ is a $b$-tree operating on a partition of size $2^t$.}\FigureName{basic}
\end{figure}

Now, how we would {\em like} to define the partition is as follows, which unfortunately does not quite work.  For each $j\in[m]$, we let the $j$th partition $\mathcal{P}_j$ have oracle $\oracle_j(i) = \enc(i)_j$. That is, we partition the universe according to the $j$th chunks of their codewords. See \Figure{basic} for an illustration, where upon an update $(i,\Delta)$ we write for each $P_j$ the name of the partition that is being updated. The key insight is that, by definition of the partition heavy hitters problem, any partition containing a heavy hitter $i\in[n]$ is itself a heavy partition in $\mathcal{P}_j$. Then, during a query, we would query each $P_j$ separately to obtain lists $L_j \subseteq[2^t]$. Thus the $P_j$ which succeed produce $L_j$ that contain the $j$th chunks of encodings of heavy hitters (plus potentially other $t$-bit chunks). Now, let us perform some wishful thinking: suppose there was only one heavy hitter $i^*$, and furthermore each $L_j$ contained exactly $\enc(i^*)_j$ and nothing else. Then we could obtain $i^*$ simply by concatenating the elements of each $L_j$ then decoding. In fact one can show that whp no more than an arbitrarily small fraction (say $1\%$) of the $P_j$ fail, implying we would still be able to obtain $99\%$ of the chunks of $\enc(i^*)$, which is sufficient to decode. The main complication is that, in general, there may be more than one heavy hitter (there may be up to $\Theta(\log n)$ of them). Thus, even if we performed wishful thinking and pretended that every $P_j$ succeeded, and furthermore that every $L_j$ contained exactly the $j$th chunks of the encodings of heavy hitters and nothing else, it is not clear how to perform the concatenation. For example, if there are two heavy hitters with encoded indices $1100$ and $0110$ with $m = t = 2$, suppose the $P_j$ return $L_1 = \{11, 01\}$ and $L_2 = \{00, 10\}$ (i.e.\ the first and second chunks of the encodings of the heavy hitters). How would we then know which chunks matched with which for concatenation? That is, are the heavy hitter encodings $1100, 0110$, or are they $1110, 0100$? Brute force trying all possibilities is too slow, since $m = \Theta(\log n/\log\log n)$ and each $|L_j|$ could be as big as $\Theta(\log n)$, yielding $(C\log n)^m = \poly(n)$ possibilities. In fact this question is quite related to the problem of {\em list-recoverable codes} (see \Section{list-recovery}), but since no explicit codes are known to exist with the efficiency guarantees we desire, we proceed in a different direction.

\newlength\Colsep
\setlength\Colsep{10pt}

\begin{figure}
\begin{center}
\noindent\begin{minipage}{3.3in}
\begin{minipage}[c][6cm][c]{\dimexpr0.25\textwidth-0.5\Colsep\relax}
\scalebox{0.25}{\includegraphics{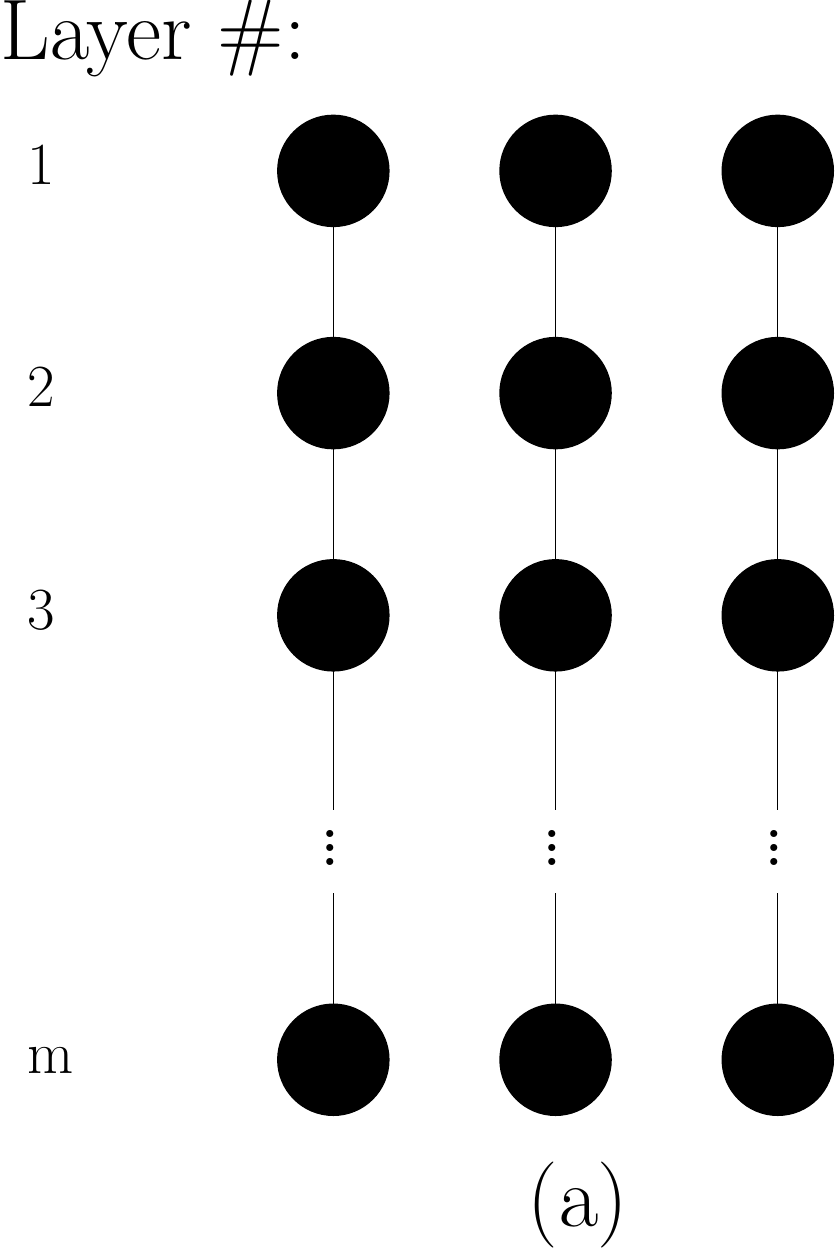}}
\end{minipage}\hfill
\hspace{.2in}\begin{minipage}[c][6cm][c]{\dimexpr0.25\textwidth-0.5\Colsep\relax}
\scalebox{0.25}{\includegraphics{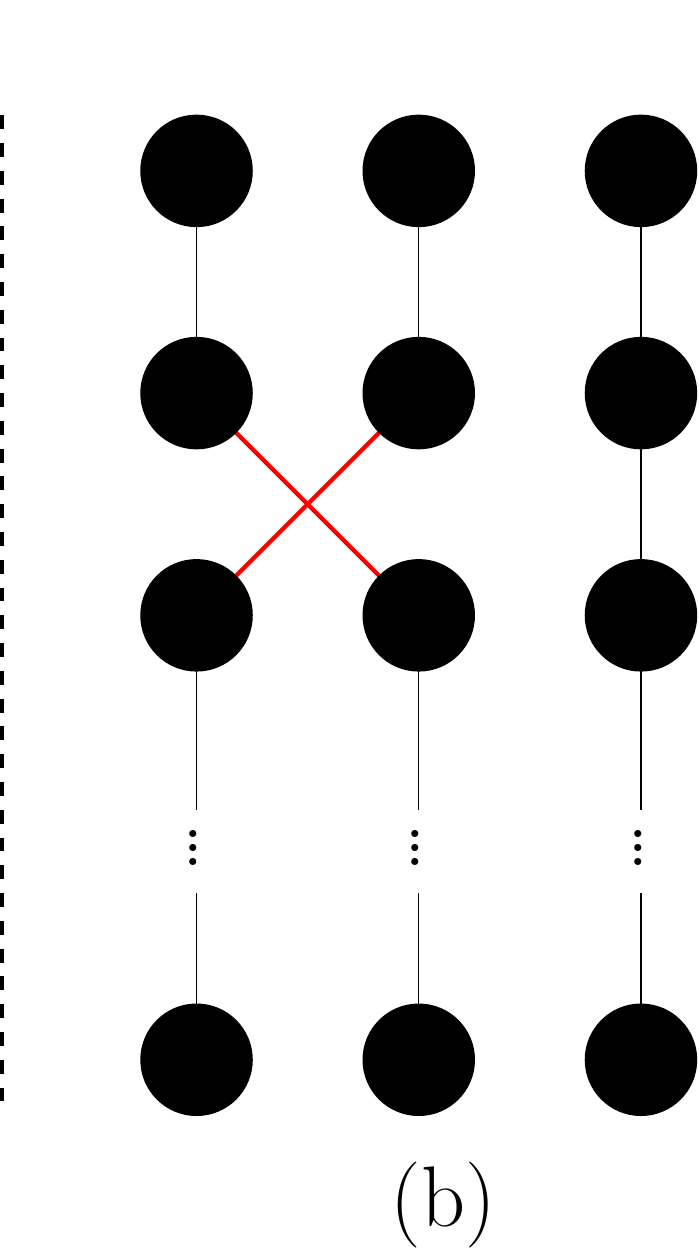}}
\end{minipage}\hfill
\begin{minipage}[c][6cm][c]{\dimexpr0.25\textwidth-0.5\Colsep\relax}
\scalebox{0.25}{\includegraphics{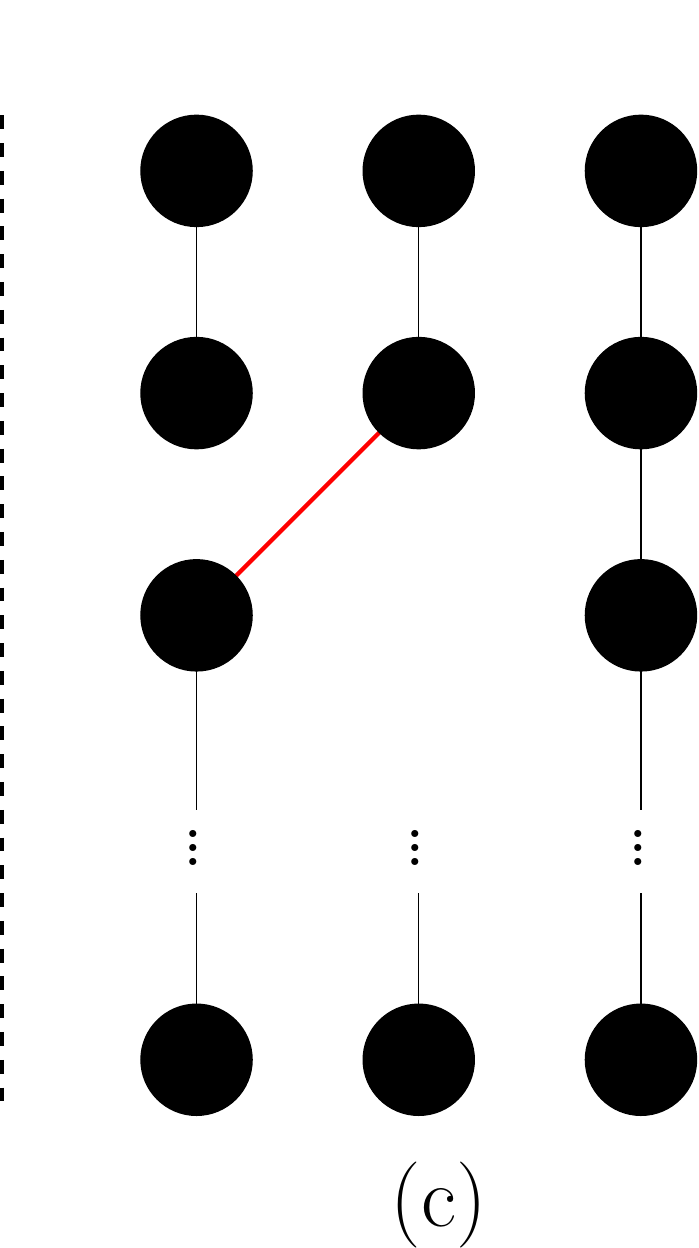}}
\end{minipage}\hfill
\begin{minipage}[c][6cm][c]{\dimexpr0.25\textwidth-0.5\Colsep\relax}
\scalebox{0.25}{\includegraphics{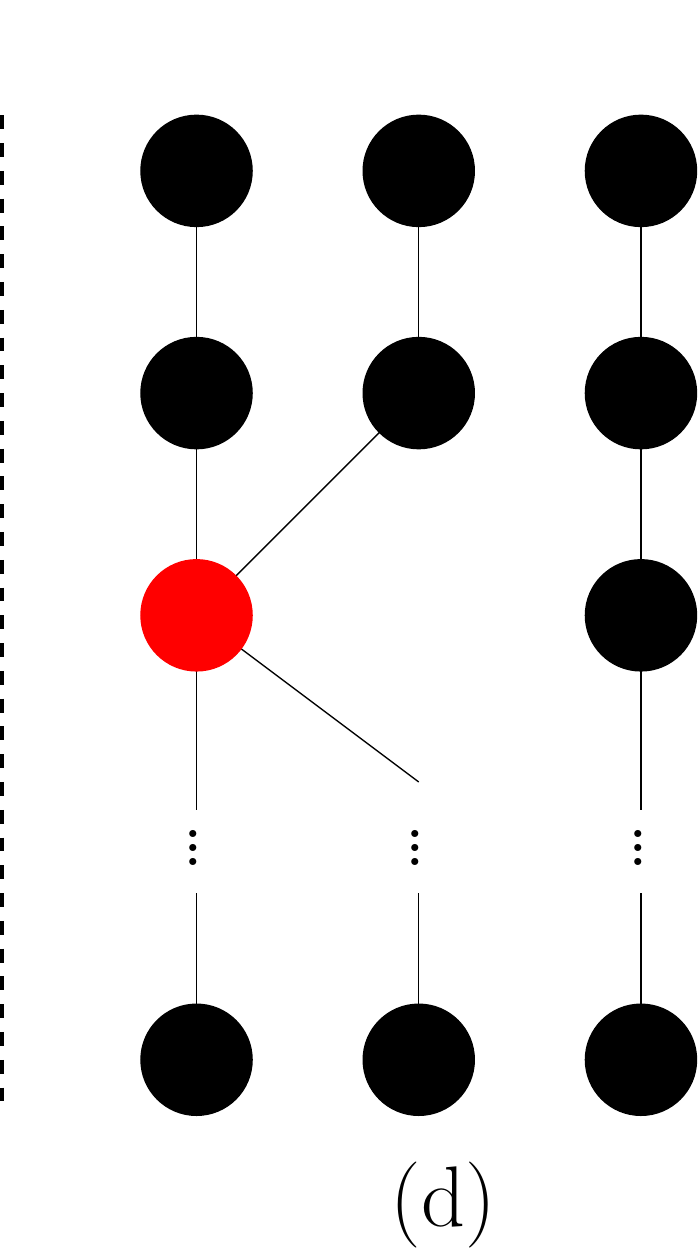}}
\end{minipage}
\end{minipage}
\end{center}
\vspace{-.5in}\caption{Each vertex in row $j$ corresponds to an element of $L_j$, i.e.\ the heavy hitter chunks output by $P_j$. When indices in $\mathcal{P}_j$ are partitioned by $h_j(i)\circ \enc(i)_j\circ h_{j+1}(i)$, we connect chunks along paths. Case (a) is the ideal case, when all $j$ are good. In (b) $P_2$ failed, producing a wrong output that triggered incorrect edge insertions. In (c) both $P_2$ and $P_3$ failed, triggering an incorrect edge and a missing vertex, respectively. In (d) two heavy hitters collided under $h_3$, causing their vertices to have the same name thereby giving the appearance of a merged vertex. Alternatively, light items masking as a heavy hitter might have appeared in $L_3$ with the same $h_3$ evaluation as a heavy hitter but different $h_4$ evaluation, causing the red vertex to have two outgoing edges to level $4$.}\FigureName{badpath}
\end{figure}

To aid us in knowing which chunks to concatenate with which across the $L_j$, a first attempt (which also does not quite work) is as follows. Define $m$ pairwise independent hash functions $h_1,\ldots,h_m:[n]\rightarrow[\poly(\log n)]$. Since there are $O(\log n)$ heavy hitters, any given $h_j$ perfectly hashes them with decent probability. Now rather than partitioning according to $\oracle_j(i) = \enc(i)_j$, we imagine setting $\oracle_j(i) = h_j(i)\circ \enc(i)_j\circ h_{j+1}(i)$ where $\circ$ denotes concatenation of bitstrings. Define an index $j\in[m]$ to be {\em good} if (a) $P_j$ succeeds, (b) $h_j$ perfectly hashes all heavy hitters $i\in[n]$, and (c) for each heavy hitter $i$, the total $\ell_2$ weight from non-heavy hitters hashing to $h_j(i)$ is $o((1/\sqrt{\log n})\|x_{\projeps}\|_2$. A simple argument shows that whp a $1-\epsilon$ fraction of the $j\in[m]$ are good, where $\epsilon$ can be made an arbitrarily small positive constant. Now let us perform some wishful thinking: if {\em all} $j\in[m]$ are good, and furthermore no non-heavy elements appear in $L_j$ with the same $h_j$ but different $h_{j+1}$ evaluation as an actual heavy hitter, then the indices in $L_j$ tell us which chunks to concatenate within $L_{j+1}$, so we can concatenate, decode, then be done. Unfortunately a small constant fraction of the $j\in[m]$ are not good, which prevents this scheme from working (see \Figure{badpath}). Indeed, in order to succeed in a query, for each heavy hitter we must correctly identify a large connected component of the vertices corresponding to that heavy hitter's path --- that would correspond to containing a large fraction of the chunks, which would allow for decoding. Unfortunately, paths are not robust in having large connected component subgraphs remaining even for $O(1)$ bad levels.

The above consideration motivates our final scheme, which uses an expander-based idea first proposed in \cite{GilbertLPS14} in the context of ``for all'' $\ell_1/\ell_1$ sparse recovery, a problem in compressed sensing. Although our precise motivation for the next step is slightly different than in \cite{GilbertLPS14}, and our query algorithm and our definition of ``robustness'' for a graph  will be completely different, the idea of connecting chunks using expander graphs is very similar to an ingredient in that work. The idea is to replace the path in the last paragraph by a graph which is robust to a small fraction of edge insertions and deletions, still allowing the identification of a large connected component in such scenarios. Expander graphs will allow us to accomplish this. For us, ``robust'' will mean that over the randomness in our algorithm, whp each corrupted expander is still a spectral cluster as defined in \Section{cluster}. For \cite{GilbertLPS14}, robustness meant that each corrupted expander still contains an induced small-diameter subgraph (in fact an expander) on a small but constant fraction of the vertices, which allowed them a recovery procedure based on a shallow breadth-first search. They then feed the output of this breadth-first search into a recovery algorithm for an existing list-recoverable code (namely Parvaresh-Vardy codes). Due to suboptimality of known list-recoverable codes, such an approach would not allow us to obtain our optimal results.

\paragraph{An expander-based approach:} Let $F$ be an arbitrary $d$-regular connected graph on the vertex set $[m]$ for some $d = O(1)$. For $j\in [m]$, let $\Gamma(j)\subset [m]$ be the set of neighbors of vertex $j$. We partition $[n]$ according to $\oracle_j(i) = z(i)_j = h_j(i)\circ \enc(i)_j \circ h_{\Gamma(j)_1}\circ\cdots \circ h_{\Gamma(j)_d}$ where $\Gamma(j)_k$ is the $k$th neighbor of $j$ in $F$. Given some such $z$, we say its {\em name} is the first $s = O(\log\log n)$ bits comprising the $h_j$ portion of the concatenation. Now, we can imagine a graph $G$ on the layered vertex set $V = [m]\times [2^s]$ with $m$ layers. If $L_j$ is the output of a heavy hitter query on $P_j$, we can view each element $z$ of $L_j$ as suggesting $d$ edges to add to $G$, where each such $z$ connects $d$ vertices in various layers of $V$ to the vertex in layer $j$ corresponding to the name of $z$. The way we actually insert the edges is as follows. First, for each $j\in[m]$ we instantiate a partition point query structure $Q_j$ as per \Lemma{pcs} with partition $\mathcal{P}_j$, failure probability $1/\poly(\log n)$, and error parameter $c\eps$ for a small constant $c$. We modify the definition of a level $j\in[m]$ being ``good'' earlier to say that $Q_j$ must also succeed on queries to every $z\in L_j$. We point query every partition $z\in L_j$ to obtain an estimate $\tilde{y}_z$ approximating $y_z$ (for the exact form of approximation, see \Equation{pcs-condition} from \Lemma{pcs}). We then group all $z\in L_j$ by name, and within each group we remove all $z$ from $L_j$ except for the one with the largest $\tilde{y}_z$, breaking ties arbitrarily. This filtering step guarantees that the vertices in layer $j$ have unique names, and furthermore, when $j$ is good all vertices corresponding to heavy hitters appear in $L_j$ and none of them are thrown out by this filtering. We then let $G$ be the graph created by including the at most $(d/2)\cdot\sum_j |L_j|$ edges suggested by the $z$'s across all $L_j$ (we only include an edge if both endpoints suggest it). Note $G$ will have many isolated vertices since only $m\cdot \max_j |L_j| = O(\log^2 n /\log\log n)$ edges are added, but the number of vertices in each layer is $2^s$, which may be a large power of $\log n$. We let $G$ be its restriction to the union of non-isolated vertices and vertices whose names match the hash value of a heavy hitter at the $m$ different levels. This ensures $G$ has $O(\log^2 n/\log\log n)$ vertices and edges. We call this $G$ the {\em chunk graph}. 

Now, the intuitive picture is that $G$ {\em should} be the vertex-disjoint union of several copies of the expander $F$, one for each heavy hitter, plus other junk edges and vertices coming from other non-heavy hitters in the $L_j$. Due to certain bad levels $j$ however, some expanders might be missing a small constant $\epsilon$-fraction of their edges, and also the $\epsilon m$ bad levels may cause spurious edges to connect these expanders to the rest of the graph. The key insight is as follows. Let $W$ be the vertices of $G$ corresponding to some particular heavy hitter, so that in the ideal case $W$ would be a single connected component whose induced graph is $F$. What one can say, even with $\epsilon m$ bad levels, is that every heavy hitter's such vertices $W$ forms an {\em $O(\epsilon)$-spectral cluster} as per \Definition{spec-cluster}. Roughly this means that (a) the cut separating $W$ from the rest of $G$ has $O(\epsilon)$ conductance (i.e.\ it is a very sparse cut), and (b) for any cut $(A,W\backslash A)$ within $W$, the number of edges crossing the cut is what is guaranteed from a spectral expander, minus $O(\epsilon)\cdot \vol(W)$.  Our task then reduces to finding all $\epsilon$-spectral clusters in a given graph. We show in \Section{cluster} that for each such cluster $W$, we are able to find a $(1-O(\epsilon))$-fraction of its volume with at most $O(\epsilon)\cdot \vol(W)$ erroneous volume from outside $W$. This suffices for decoding for $\epsilon$ a sufficiently small constant, since this means we find most vertices, i.e.\ chunks of the encoding, of the heavy hitter. 

For the special case of $\ell_1$ heavy hitters in the strict turnstile model, we are able to devise a much simpler query algorithm that works; see \Section{strict} for details. For this special case we also give in \Section{expected-time} a space-optimal algorithm with $O(\log n)$ update time, whp success, and {\em expected} query time $O(\eps^{-1}\log n)$ (though unfortunately the variance of the query time may be quite high).

\subsection{Cluster-preserving clustering}\SectionName{clustering-overview}
An $\epsilon$-spectral cluster is a subset $W$ of the vertices in a graph $G = (V, E)$ such that (1) $|\partial W| \le \epsilon\vol(W)$, and (2) for any $A\subsetneq W$ with $\vol(A)/\vol(W) = r$, $|E(A, W\backslash A)| \ge (r(1-r) - \epsilon)\vol(W)$. Item (2) means the number of edges crossing a cut within $W$ is what you would expect from a random graph, up to $\epsilon \vol(W)$. Our goal is to, given $G$, find a partition of $V$ such that every $\epsilon$-spectral cluster $W$ in $G$ matches some partition up to $\epsilon\vol(W)$ symmetric difference. 

\begin{figure}
  \begin{minipage}{5.0in}
\begin{minipage}[c][6cm][c]{\dimexpr0.25\textwidth-0.5\Colsep\relax}
\textcolor{white}{move to the right hi there again}\scalebox{0.4}{\includegraphics{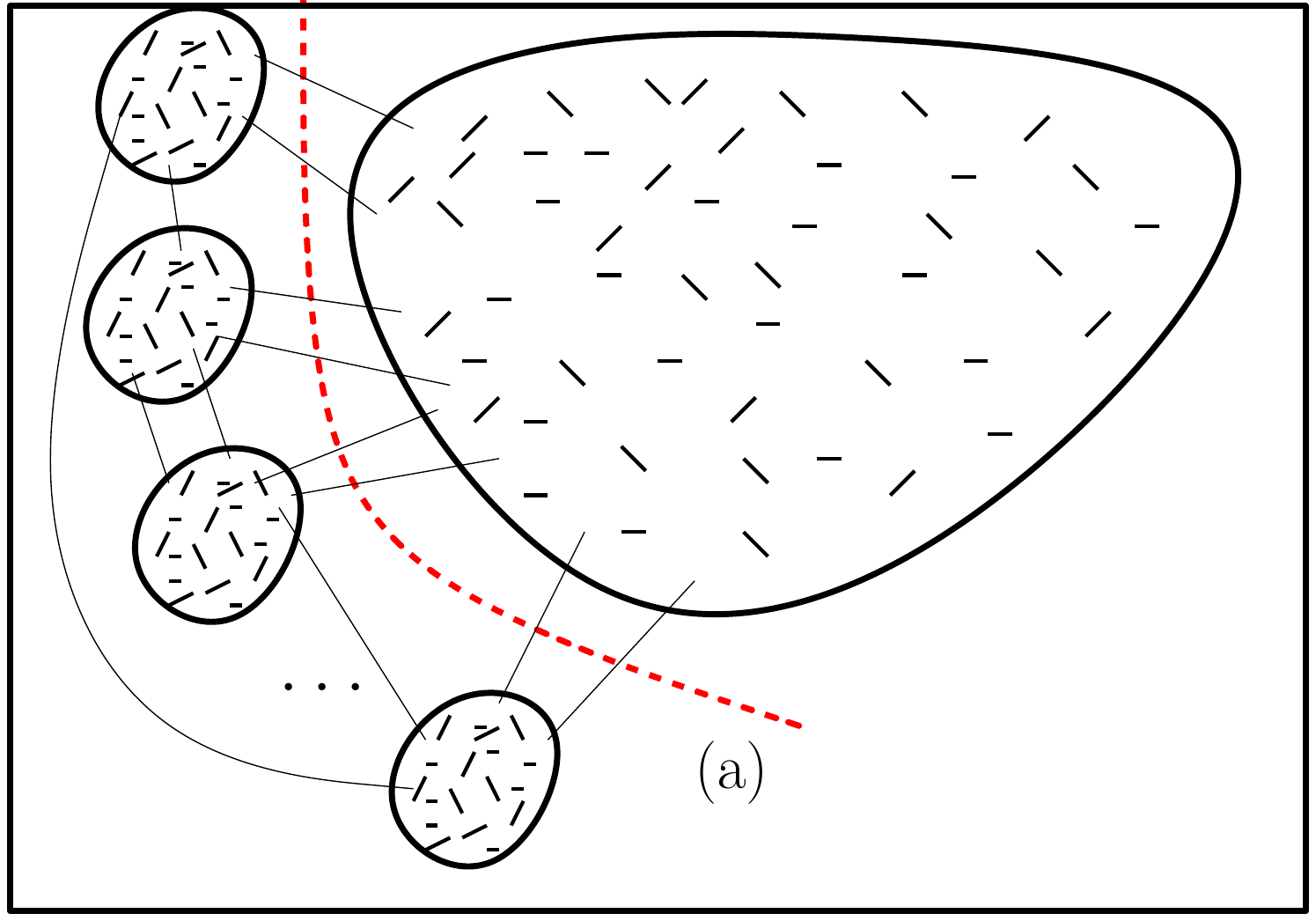}}
\end{minipage}\hfill
\begin{minipage}[c][6cm][c]{\dimexpr0.25\textwidth-0.5\Colsep\relax}
\textcolor{white}{move to the right hi there again}\scalebox{0.4}{\includegraphics{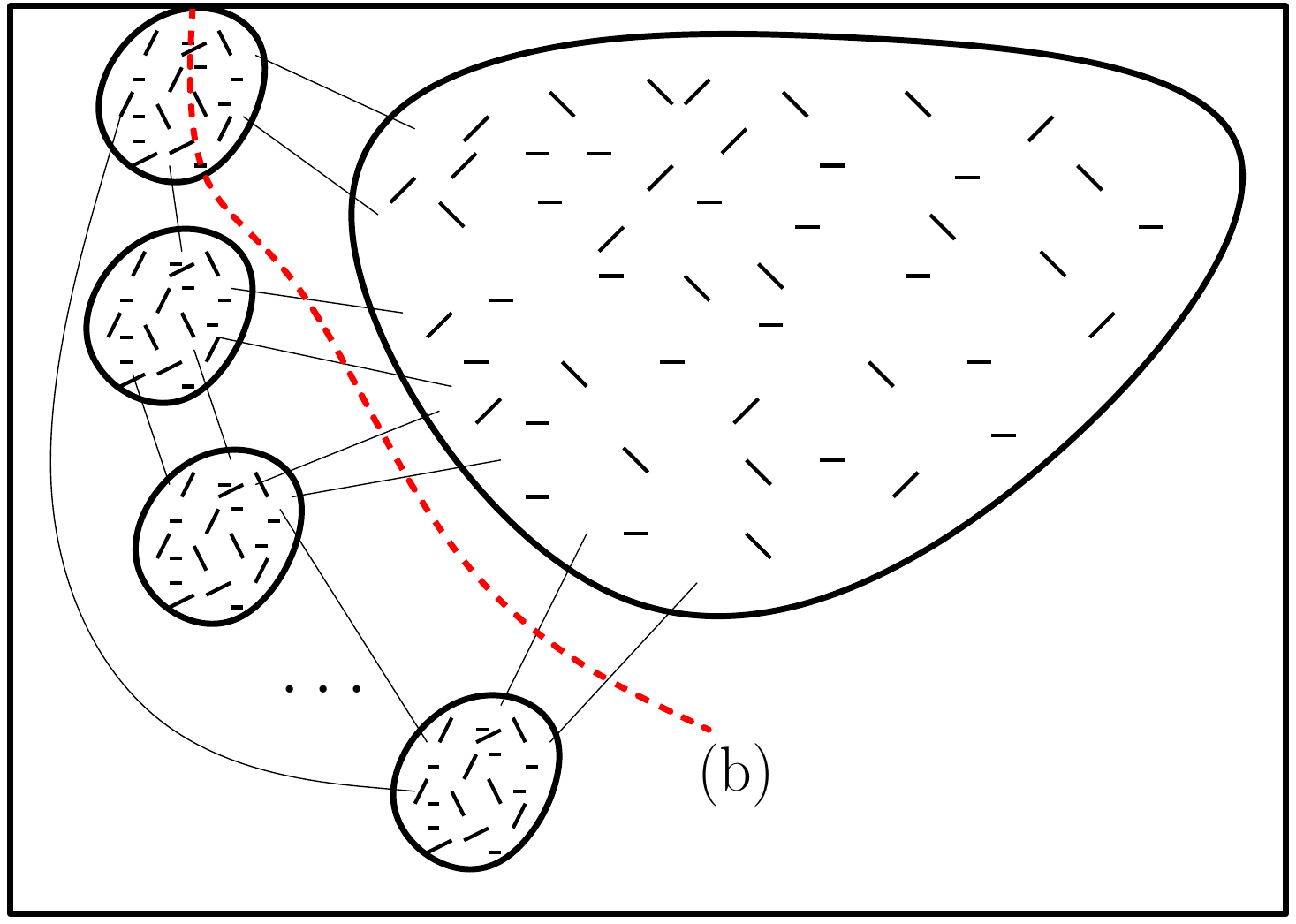}}
\end{minipage}
\end{minipage}
\vspace{-.1in}\caption{Each small oval is a spectral cluster. They are well-connected internally, with sparse cuts to the outside. The large oval is the rest of the graph, which can look like anything. Cut (a) represents a good low-conductance cut, which makes much progress (cutting the graph in roughly half) while not losing any mass from any cluster. Cut (b) is also a low-conductance cut as long as the number of small ovals is large, since then cutting one cluster in half has negligible effect on the cut's conductance. However, (b) is problematic since recursing on both sides loses half of one cluster forever.}\FigureName{badcut}
\end{figure}

Our algorithm \cutandclose is somewhat similar to the spectral clustering algorithm of \cite[Section 4]{KannanVV04}, but with local search. That algorithm is quite simple: find a low-conductance cut (e.g.\ a Fiedler cut) to split $G$ into two pieces, then recurse on both pieces. Details aside, Fiedler cuts are guaranteed by Cheeger's inequality to find a cut of conductance $O(\sqrt{\gamma})$ as long as a cut of conductance at most $\gamma$ exists in the graph. The problem with this basic recursive approach is shown in \Figure{badcut} (in particular cut (b)). Note that a cluster can be completely segmented after a few levels of recursion, so that a large portion of the cluster is never found.

Our approach is as follows. Like the above, we find a low-conductance cut then recurse on both sides. However, before recursing on both sides we make certain ``improvements'' to the cut. We say $A\subset V$ is {\em closed} in $G$ if there is no vertex $v\in G\backslash A$ with at least $5/9$ths of its neighbors in $A$. Our algorithm maintains that all recursive subcalls are to closed subsets in $G$ as follows. Suppose we are calling \cutandclose on some set $A$. We first try to find a low-conductance cut within $A$. If we do not find one, we terminate and let $A$ be one of the sets in the partition. Otherwise, if we cut $A$ into $(S,\bar{S})$, then we close both $S, \bar{S}$ by finding vertices violating closure and simply moving them. It can be shown that if the $(S, \bar{S})$ cut had sufficiently low conductance, then these local moves can only improve conductance further. Now both $S$ and $\bar{S}$ are closed in $A$ (which by a transitivity lemma we show, implies they are closed in $G$ as well). We then show that if (1) some set $S$ is closed, and (2) $S$ has much more than half the volume of some spectral cluster $W$ (e.g.\ a $2/3$rds fraction), then in fact $S$ contains a $(1-O(\epsilon))$-fraction of $W$. Thus after closing both $S, \bar{S}$, we have that $S$ either: (a) has almost none of $W$, (b) has almost all of $W$, or (c) has roughly half of $W$ (between $1/3$ and $2/3$, say). To fix the latter case, we then ``grab'' all vertices in $\bar{S}$ with some $\Omega(1)$-fraction, e.g. $1/6$th, of their neighbors in $S$ and simply move them all to $S$. Doing this some constant number of times implies $S$ has much more than $2/3$rds of $W$ (and if $S$ was in case (a), then we show it still has almost none of $W$). Then by doing another round of closure moves, one can ensure that both $S, \bar{S}$ are closed, and each of them has either an $O(\epsilon)$-fraction of $W$ or a $(1-O(\epsilon))$-fraction. The details are in \Section{cluster}. It is worth noting that our algorithm can make use of {\em any} spectral cutting algorithm as a black box and not just Fiedler cuts, followed by our grab and closure steps. For example, algorithms from \cite{OrecchiaV11,OrecchiaSV12} run in nearly linear time and either (1) report that no $\gamma$-conductance cut exists (in which case we could terminate), (2) find a {\em balanced} cut of conductance $O(\sqrt{\gamma})$ (where both sides have nearly equal volume), or (3) find an $O(\sqrt{\gamma})$-conductance cut in which every $W\subset G$ with $\vol(W) \le (1/2)\vol(G)$ and $\phi(W) \le O(\gamma)$ has more than half its volume on the smaller side of the cut. Item (2), if it always occurred, would give a divide-and-conquer recurrence to yield nearly linear time for finding all clusters. It turns out item (3) though is even better! If the small side of the cut has half of every cluster $W$, then by grabs and closure moves we could ensure it is still small and has almost all of $W$, so we could recurse just on the smaller side. It appears an $O(|E|\log|V|)$-space implementation of such a cutting algorithm achieving this guarantee would lead to $O(\eps^{-2}\log^{2+o(1)} n)$ query time for whp heavy hitters (with $O(\eps^{-2}\log^{1+o(1)} n)$ expected query time), but in this version of our work we simply focus on achieving $O(\eps^{-2}\poly(\log n))$.

%% file: turnstile.tex
\section{General turnstile updates}\SectionName{turnstile}
In this section we analyze our algorithm \es described in \Section{overview}. Recall that our final algorithm, including the reduction in \Section{reduction}, is as follows. First we pick a hash function $h:[n]\rightarrow[q]$ from a $\Theta(\log n)$-wise independent hash family for $q = \ceil{1/(\eps^2\log n)}$. Then we initialize $q$ data structures $D^1,\ldots,D^q$, where each $D^k$ is an $\eps'$-heavy hitters data structure as described in \Section{overview}, for $\eps'= \max\{\eps, 1/\sqrt{C\log n}\}$. We also during intialization construct a $d$-regular $\lambda_0$-spectral expander $F$ on $m = \Theta(\log n / \log\log n)$ vertices for some $d = O(1)$, where $\lambda_0 = \epsilon d$ for some (small) constant $\epsilon>0$ to be specified later. Such an $F$ can be constructed in time $\poly(\log n)$ with $d = \poly(1/\epsilon)$ deterministically \cite{ReingoldVW02}, then stored in adjacency list representation consuming $O(\log n / \log\log n)$ words of memory.

The $D^k$ are independent, and for each $D^k$ we pick hash functions $h_1^k,\ldots,h_m^k:[n]\rightarrow[\poly(\log n)]$ independently from a pairwise independent family and instantiate partition $\eps'$-heavy hitter data structures $P_1^k,\ldots,P_m^k$ (as $b$-trees; see \Section{btree}) with $\oracle_j^k:[n]\rightarrow [2^{O(t)}]$ for $j\in[m]$ defined by $\oracle_j^k(i) = h_j^k(i)\circ \enc(i)_j \circ h^k_{\Gamma(j)_1}\circ\cdots \circ h^k_{\Gamma(j)_d}$. Each $P_j^k$ has failure probability $1/\poly(\log n)$. We also instantiate partition $c\eps$-point query data structures $Q_1^k,\ldots,Q_m^k$ with small constant $c$ and failure probability $1/\poly(\log n)$ as per \Lemma{pcs}, with the same $\oracle_j^k$. Here $\enc$ is an encoding as in \Section{overview} mapping into $T = O(\log n)$ bits, and $t = T/m = \Theta(\log\log n)$. $\Gamma(j)_k$ is the $k$th neighbor of $j$ in $F$. To answer a query to our overall data structure, we query each $D^k$ separately then output the union of their results. To answer a query to $D^k$, we form a chunk graph $G^k$ from the outputs of the $P^k_j$ as in \Section{overview}. We then find all $\epsilon_0$-spectral clusters using \Theorem{cluster} for a sufficiently small constant $\epsilon_0$, throw away all clusters of size less than $m/2$, then from each remaining cluster $W'$ we: (1) remove all vertices of degree $\le d/2$, (2) remove $v\in W'$ coming from the same layer $j\in[m]$ as some other $v'\in W'$, then (3) form a (partially corrupted) codeword using the bits associated with vertices left in $W'$. We then decode to obtain a set $B^k$ containing all heavy hitter indices $i\in h^{-1}(k)$ with high probability. We then output $L = \cup_{k=1}^q B^k$ as our final result.

Henceforth we condition on the event $\mathcal{E}$ that every heavy hitter $i\in[n]$ is $\eps'$-heavy in $D^{h(i)}$, which happens whp by \Theorem{reduce-threshold}. Our final analysis makes use of the following.

\begin{lemma}\LemmaName{bad-levels}
Suppose $\mathcal{E}$ occurs. Focus on a particular $k\in[q]$. Let $x' = x_{h^{-1}(k)}$ be the projection of $x$ onto vertices hashing to $k$. As in \Section{overview}, we say an index $j\in[m]$ is {\em good} if (a) $P^k_j$ succeeded, (b) $h_j^k$ perfectly hashes all $\eps'$-heavy hitters in $x'$, (c) for each $\eps'$-heavy hitter $i$, the total $\ell_2$ weight from non-heavy hitters in $h^{-1}(k)$ hashing to $h_j(i)$ is under $h_j$ is $o((1/\sqrt{\log n}))\|x_{\projepstwo{\eps'}}\|_2$, and (d) $Q^k_j$ succeeded on every $z\in L^k_j$. Call $j$ {\em bad} otherwise. Then if the failure probability of each $P^k_j, Q^k_j$ is $1/\log^{C+1} n$, the range of each $h_j^k$ is $[\log^{C+3} n]$, and $m = C\log n/\log\log n$ for some constant $C>0$, then with probability $1 - 1/n^c$ the number of bad levels is at most $\beta m$, where $c$ can be made an arbitrarily large constant and $\beta$ an arbitrarily small constant by increasing $C$.
\end{lemma}
\begin{proof}
The probability of (b), that $h_j^k$ perfectly hashes all $O(\log n)$ $\eps'$-heavy hitters, is $1 - O(1/\log^{C+1} n)$. For (c), focus on a particular heavy hitter $i$. The expected squared $\ell_2$ mass from non-$\eps'$ heavy hitters colliding $i$ under $h_j$ is $(1/\log^{C+3} n)\|x_{\projepstwo{\eps'}}\|_2^2$ by pairwise independence of $h_j$. Thus by Markov, the probability that more than $(1/\log^2 n)\|x_{\projepstwo{\eps'}}\|_2^2$ squared $\ell_2$ mass collided with $i$ from light items is at most $1/\log^{C+1} n$. Since there are only $O(\log n)$ $\eps'$-heavy hitters, by a union bound we have that (c) holds, i.e. that no heavy hitter collides with more than $(1/\log n)\|x_{\projepstwo{\eps'}}\|_2 = o(1/\sqrt{\log n})\|x_{\projepstwo{\eps'}}\|_2$ $\ell_2$ mass from light items under $\ell_j$, with probability $1 - O(1/\log^C n)$. For $(d)$, we perform $O(1/\eps'^2) = O(\log n)$ point queries on $Q_j^k$, and each point query has failure probability $1/\log^{C+1} n$, so the probability that any of the point query fails is at most $1/\log^C n$. Thus $j$ is bad with probability $O(1/\log^C n)$. Now also choose $m = C\log n/\log\log n$. Then the probability that there are more than $\beta m$ bad levels is at most, for some constant $C'$ independent of $C$,
$$
\binom{m}{\beta m} \left(\frac {C'}{\log^C n}\right)^{\beta m} \le \left(\frac{C'e}{\beta\log^C n}\right)^{\beta C\log n/\log\log n} = \left(\frac{C'e}{\beta}\right)^{\beta C\log n/\log\log n}\cdot \frac 1{n^{\beta C^2}} = \frac 1{n^{\beta C^2 - o_n(1)}},
$$
using $\binom{a}{b} \le (ea/b)^b$. Thus the lemma holds for $\beta = 1/C$, and $c = C/2$ for sufficiently large $n$.
\end{proof}

Our correctness analysis also makes use of the following lemma \cite[Lemma 2.3]{AlonC06} to show that, with high probability, each heavy hitter is represented by an $\epsilon_0$-spectral cluster. The lemma states that for spectral expanders, sets expand with small sets $S$ satisfying even better expansion properties. Some version of this lemma is known in the literature as the expander mixing lemma.

\begin{lemma}{{\cite{AlonC06}}}\LemmaName{alonc}
Let $A$ be the adjacency matrix of a $d$-regular graph with vertex set $V$. Suppose the second largest eigenvalue of $A$ in magnitude is $\lambda>0$. Then for any $S\subseteq V$, writing $|S| = r |V|$, $|\partial S| \ge (d - \lambda)(1-r)|S|$.
\end{lemma}

\begin{theorem}\TheoremName{clustering-reduction}
Assume $n$ is larger than some constant. Let $\mathcal{A}$ be the algorithm from \Theorem{cluster} which finds all $\epsilon_0$-spectral clusters in a graph on $O(\eps'^{-2}\log n / \log\log n)$ vertices and edges whp in time $\mathcal{T}$ and space $\mathcal{S}$. For any $0<\eps < 1/2$, there is an algorithm solving the $\eps$-heavy hitters problem in general turnstile streams on vectors of length $n$ whp with space and update time $O(\eps^{-2}\log n)$ and $O(\log n)$, respectively. Answering a query uses an additional $O(\mathcal{S}) = o(\eps^{-2}\log n)$ space and runs in time $O(\eps^{-2}\log^{1+\gamma} n + q\mathcal{T}) = O(\eps^{-2}\poly(\log n))$, by $q$ successive calls to $\mathcal{A}$. Here $\gamma>0$ can be chosen as an arbitrarily small constant.
\end{theorem}
\begin{proof}
We first analyze space and running times. We store $F$, taking $O(\log n/\log\log n)$ space, which can be ignored as it is dominated by other parts of the algorithm. Storing $h$ requires $O(\log n)$ words of memory. Let us now focus on a specific $D^k$. It stores $m$ hash functions drawn independently from a pairwise family mapping $[n]$ to $[\poly(\log n)]$, consuming $O(m) = O(\log n/\log\log n)$ space. The $m$ $b$-trees combined $P_j^k$ consume space $m\cdot O(\eps'^{-2} t) = O(\eps'^{-2}\log n)$, and the same for the $Q_j^k$. Thus the total space per $D^k$ is $O(\eps'^{-2}\log n)$, and thus the total space across the entire algorithm is $O(q\eps'^{-2}\log n) = O(\eps^{-2}\log n)$. We also need an additional $O(\mathcal{S})$ space to run $\mathcal{A}$ repeatedly during a query. For the update time, we first hash using $h$, taking $O(\log n)$ time. Then within a $D^k$ we perform $m$ hash evaluations $h_j^k$, taking total time $O(m)$. We then compute the encoding $\enc(i)$ of $i$, taking $O(\log n)$ time, then perform $m$ $b$-tree updates taking total time $m\cdot O(t) = O(\log n)$. Thus the total update time is $O(\log n)$. For the query time, for a given $D^k$ we have to query $m$ $b$-trees, taking time $m\cdot O(\eps'^{-2}t \log^{\gamma}n) = O(\eps'^{-2}\log^{1+\gamma}n)$. We also run $\mathcal{A}$ on $G^k$, taking time $\mathcal{T}$. Thus the total query time is $O( q(\eps'^{-2}\log^{1+\gamma}n + \mathcal{T})) = O(\eps^{-2}\log^{1+\gamma}n + q\mathcal{T})$.

Now it only remains to argue correctness. Each $L^k$ has size $O(1/\eps'^2)$ with probability $1$, which is a guarantee of the $b$-tree. Then when we form the chunk graph $G^k$, it has $O(\eps'^{-2}\log n/\log\log n)$ vertices. We only insert a decoded spectral cluster into $L$ if its size is at least $m/2$, so there can be at most $O(1/\eps'^2)$ such clusters found. Thus $|L| \le O(q/\eps'^2) = O(1/\eps^2)$ with probability $1$.

It now only remains to show that each heavy hitter $i\in[n]$ is represented by some spectral cluster of size at least $m/2$ in $G^{h(i)}$. We assume event $\mathcal{E}$, which happens whp.  Now let us focus on a particular $G^k$. Define $x' = x_{h^{-1}(k)}$. We will show that each $\eps'$-heavy hitter in $x'$ is represented by such a cluster in $G^k$. Focus on a particular $\eps'$-heavy hitter $i$. Let $W$ be the set of $m$ vertices $\{(j, h_j^k(i))\}_{j=1}^m$. Suppose that all levels $j\in[m]$ were good. We claim that in this case, $W$ would be an isolated connected component in $G^k$, and furthermore the induced graph on $W$ would be $F$. To see this, observe that before filtering by $Q^k_j$, $z^k(i)_j$ appears in $L^k_j$ (see \Section{overview} for the definitions of $L_j$ and $z(i)_j$; $L^k_j$, $z^k(i)_j$ are then the natural modifications corresponding to $k\in[q]$). Once we group the $z\in L^k_j$ by name and filter according to $Q^k_j$, let $z\neq z^k(i)_j$ be such that its name also equals $h_j^k(i)$. Since level $j$ is good, none of the mass of $z$ is from a heavy hitter. Thus by condition (c) in the definition of goodness, the $\ell_2$ weight $y_z$ of partition $z$ is at most $o(1/\sqrt{\log n})\|x_{\projepstwo{\eps'}}\|_2$. Thus by the condition of \Equation{pcs-condition}, since $Q^k_j$ succeeds we will will have $\tilde{y}_{z^k(i)_j} > \tilde{y}_z$ and not remove $z^k(i)_j$ from $L^k_j$. Thus $z^k(i)_j$ remains in $L^k_j$ even after filtering. Since this is true across all $j\in[m]$, we will add all edges of $F$ to $W$. Furthermore no vertices outside $W$ will connect to $W$, since edge insertion requires mutual suggestion of the edge by both endpoints. Thus, in the case of all good levels, by \Lemma{alonc} for any subset $A$ of $W$ with $|A| = r|W| = rm$, 
$$
|E(A, W\backslash A)| \ge (d - \lambda_0)r(1-r)|W| = (r(1-r) - r(1-r)\frac{\lambda_0}d)d|W|\ge (r(1-r) - \frac{\lambda_0}{4d})dm.
$$

By \Lemma{bad-levels}, the number of bad levels is at most $\beta m$ with high probability, for $\beta$ an arbitrarily small constant. Let us now understand the possible effects of a bad level on $W$. Suppose $j$ is bad. Then the $L^k_j$ obtained after filtering by $Q^k_j$ may not include $z^k(i)_j$. Thus $W$ may lose at most $d$ edges from $F$ (those corresponding to edges incident upon vertex $j$ in $F$). Second, $L^k_j$ after filtering may instead contain some $z$ whose name is $h^k_j(i)$, but whose edge suggestions differ from $z^k(i)_j$, possibly causing edges to be inserted within $W$ which do not agree with $F$, or even inserting edges that cross the cut $(W, G\backslash W)$. At most $d$ such edges are inserted in this way. Thus, across all bad levels, the total number of $F$-edges deleted within $W$ is at most $\beta dm$, and the total number of edges crossing the cut $(W, G\backslash W)$ is also at most $\beta dm$. Also, the volume $\vol(W)$ is always at most $dm$ and at least $(1-\beta)dm$. Thus after considering bad levels, for any subset $A$ of $W$ as above,
$$
|E(A, W\backslash A)| \ge (r(1-r) - \frac{\lambda_0}{4d} - \beta)dm \ge (r(1-r) - \epsilon_0)\vol(W).
$$
for $\epsilon_0 \ge \beta + \lambda_0/(4d)$. Furthermore, the number of edges leaving $W$ to $G\backslash W$ is
$$|\partial(W)| \le \beta dm \le \frac{\beta}{1-\beta} \vol(W) \le \epsilon_0 \vol(W)$$ 
for $\epsilon_0 \ge \beta/(1-\beta)$. Thus $W$ is an $\epsilon_0$-spectral cluster in $G^k$ representing $i$ for $\epsilon_0 = \max\{\beta + \lambda_0/(4d), \beta/(1-\beta)\} = \max\{\beta + \epsilon/4, \beta/(1-\beta)\}$. Thus by \Theorem{cluster}, $\mathcal{A}$ recovers a $W'$ missing at most $3\epsilon_0 \vol(W)$ volume from $W$, and containing at most $2250000\epsilon_0$ volume from $G\backslash W$. Note we remove any vertex from $W'$ of degree $\le d/2$, so $W'$ contains at most $5500000\epsilon_0 m$ vertices from outside $W$. Furthermore, since there are at most $\beta dm$ edges lost from $W$ due to bad levels, at most $2\beta m$ vertices in $W$ had their degrees reduced to $\le d/2$, and thus removing low-degree vertices removed at most $2\beta m$ additional vertices from $W$. Also, since the max degree in $G^k$ is $d$ and at most a $2\beta$ fraction of vertices in $W$ have $\le d/2$ degree, since $W'$ is missing at most $3\epsilon_0 \vol(W)$ vertices from $W$, this corresponds to at most $2\beta m + 6\epsilon_0 m$ vertices missing from $W$. We then form a (corrupted) codeword $C'\in \{0,1\}^{mt}$ by concatenating encoding chunks specified by the vertices in $W'$. Since then at most a $(5500006\epsilon_0 + 2\beta)$-fraction of entries in $\enc(i)$ and $C'$ differ, for $\epsilon, \beta$ sufficiently small constants we successfully decode and obtain the binary encoding of $i$, which is included in $L$.
\end{proof}

\begin{remark}\RemarkName{recurse}
\textup{
It is possible to slightly modify our algorithm so that the $\eps^{-2}\log^{1+\gamma} n$ summand in the query time becomes $\eps^{-2}\log^{1+o(1)} n$. We state the modification here without providing the full details. Specifically, rather than implement the $P_j^k$ as $b$-trees, one could instead implement them as recursive versions of the \es on a smaller universe. Then, after a single level of recursion down,  one switches back to using $b$-trees. Then the $b$-trees at the second level of recursion have query time which includes an arbitrarily small power of $\log\log n$, which is $\log^{o(1)} n$. One may be tempted to avoid $b$-trees altogether and simply continue the recursion to the leaves, but due to the blowup in bitlength of the partition sizes stemming from concatenating $h_j$ values, it seems such an approach would result in the space complexity being multiplied by an additional $2^{O(\log^* n)}$ factor.
}
\end{remark}

%% file: cluster.tex
\renewcommand{\eps}{\epsilon}
\section{Cluster-preserving clustering}\SectionName{cluster}
In this section we present our cluster-preserving clustering, partitioning
the graph into clusters with the promise to preserve and
identify all $\epsilon$-spectral clusters as per \Definition{spec-cluster}.

The main goal of this section is to prove Theorem \ref{thm:cluster}. To prove Theorem \ref{thm:cluster}, we shall use cuts based
on Cheeger's inequality. Given an undirected graph $G = (V, E)$, let its Laplacian matrix be $\lap_G$ (normalized so that $(\lap_G)_{i,i} = 1$ for all $i$). It is known that that $\lap_G$ is positive semidefinite with eigenvalues $0 = \lambda_1 \le \lambda_2 \le \ldots \le \lambda_{|V|}$ with $\lambda_k = 0$ iff $G$ has at least $k$ connected components. A fundamental result in spectral graph theory is the following, which is a robust form of this fact when $k=2$.

\begin{theorem}[Cheeger's inequality for graphs {\cite{AlonM85,Alon86,SinclairJ89}}]\TheoremName{cheeger}
For $G = (V,E)$ an undirected graph, let $\lambda_2$ be the second
smallest eigenvalue of $\lap_G$. Then a cut $(S,\widebar{S})$ satisfying
\begin{equation}
\frac{\lambda_2}2 \leq \phi(G)\le \phi(G, S) \le \sqrt{2\lambda_2}\le 2\sqrt{\phi(G)}
\EquationName{cheeger}
\end{equation}
can be found in time $\poly(|V|)$ and space $O(|V|+|E|)$.
\end{theorem}

\begin{algorithm}[H]
\caption{Top level of cluster preserving clustering of $G_0$ as stated in
Theorem \ref{thm:cluster}.}\label{alg:main}
\begin{algorithmic}[1]
\Function{Main}{$G_0$}
\State $\{U_1\ldots,U_\ell\}\gets$ \Call{\cutandclose}{$G_0$}
\State For each $U_i$, recursively remove $v\in U_i$ with
$\geq 5/9$ths neigbors outside $U_i$.\Comment Cleaning
\State \Return $\{U_1\ldots,U_\ell\}$
\EndFunction
\end{algorithmic}
\end{algorithm}
\begin{algorithm}[H]
\caption{Cluster preserving clustering of $G=G_0|V$, approximately 
isolating clusters $\cluster$ of $G_0$.}\label{alg:cluster}
\begin{algorithmic}[1]
\Function{\cutandclose}{$G = G_0|V$}
\State Use \Theorem{cheeger} to find a cut $(S,\widebar{S})$ satisfying \Equation{cheeger}; say $|S| \le |\widebar{S}|$.
\If {$\phi(G, S)\geq 1/500$} \Return $\{V\}$ \Comment We've identified a potential cluster.
\Else
   \State \Call{LocalImprovements}{$G,S,V$} \Comment Local improvements of both $S,\widebar{S}$.
   \State \Call{Grab}{$G,S$}
  \State \Call{LocalImprovements}{$G,S,\widebar{S}$} \Comment Local improvements of $S$.    \State \Call{Grab}{$G,S$}
   \State \Call{LocalImprovements}{$G,S,\widebar{S}$} \Comment Local improvements of $S$.
   \State \Call{LocalImprovements}{$G,\widebar S, S$} \Comment Local improvements of $\widebar{S}$.
    \State $\mathcal{C} \gets $\Call{\cutandclose}{$G|S$} $\cup$ \Call{\cutandclose}{$G|\widebar{S}$}
    \State \Return $\mathcal{C}$
\EndIf
\EndFunction
\end{algorithmic}
\end{algorithm}
\begin{algorithm}
\caption{Local improvements across cut $(S,\widebar{S})$ only moving vertices from $T$.}\label{alg:improve}
\begin{algorithmic}[1]
\Function{LocalImprovements}{$G,S,T$}\Comment $(S,\widebar{S})$ is a cut in $G = (V,E)$
    \While {$\exists v\in T$ with at least $5/9$ths of its edges crossing the cut}\Comment Local improvement.
        \State Move $v$ to the other side of the cut.
    \EndWhile
\EndFunction
\end{algorithmic}
\end{algorithm}

\begin{algorithm}
\caption{Expanding $S$ by grabbing 
all vertices with $1/6$th neighbors in $S$.}\label{alg:grab}
\begin{algorithmic}[1]
\Function{Grab}{$G,S$}
\State Let $T$ be set of vertices from $G\setminus V$ with at least $1/6$th of its 
neighbors in $S$.
\State $S\gets S\cup T$.
\EndFunction
\end{algorithmic}
\end{algorithm}

We are now ready to describe our algorithm, presented in pseudo-code
in Algorithms \ref{alg:cluster}--\ref{alg:grab}.

\newcommand{\clusteralg}{\textsc{CutGrabClose}\xspace}
\newcommand{\main}{\textsc{Main}\xspace}

At the top level, from \main, the function \clusteralg is fed as 
input the graph $G=(V,E)$
from Theorem \ref{thm:cluster}. We denote this top level graph $G_0=(V_0,E_0)$.
\clusteralg is a recursive algorithm, working on induced subgraphs $G=G_0|V$
of $G_0$, producing a family of disjoint subsets of the vertices.
Suppose $G_0$ has an $\eps$-spectral cluster $\cluster$. Our goal is to show
that one set in the final partition
produced by \clusteralg matches $\cluster$ in a slightly weaker sense than Theorem~\ref{thm:cluster} (namely $\vol_{U_i} (U_i \setminus \cluster) = O(\eps \vol(W))$ instead of the second matching condition).
At the end of the section, we describe how the last cleaning step in \main removes extraneous mass from candidate clusters and achieves the stronger conditions in Theorem~\ref{thm:cluster}.

On input $G=(V,E)$, \clusteralg works as follows: if there is no
low-conductance cut, it returns just the single-element partition
$\{V\}$. Otherwise, it finds a low-conductance cut $(S,\widebar{S})$. It
then performs a sequence of recursive {\em local improvements} and two {\em
grabs} to maintain certain invariants in our analysis. Namely, we want
to maintain the invariant that throughout all levels of recursion, 
most of $\cluster$ stays in the same recursive branch. Here, by most
of $\cluster$, we mean relative to the volume of $\cluster$ in the original $G_0$.
This will ensure that the final output $\mathcal{C}$ of the topmost level
of recursion contains a single set $S$ matching $\cluster$.

We now proceed with a formal description. Since we will be talking both
about the original graph $G_0$ and the recursive induced subgraphs $G=G_0|V$,
we will sometimes use a subscript to indicate which graph we are working in,
e.g., if $A\subseteq V$, then $\vol_{G}(A)$ and $\vol_{G_0}(A)$ are
the volumes of $A$ in $G$ and $G_0$ respectively. Since the subgraphs
are all induced, the edge sets are always determined from the vertices included.

\paragraph{Closure}
We say set $A$ of vertices from $G$ 
is {\em closed\/} in $G$ if there is no vertex $v\in G\backslash A$ with at least 
$5/9$ of its neighbors in $A$. The closure property is transitive, which is
very useful for our recursion:

\begin{lemma}\label{lem:transtive}
If $A\subseteq V\subseteq V(G_0)$, $A$ is closed in $G_0|V$, and $V$ is
closed in $G_0$, then $A$ is closed in $G_0$.
\end{lemma}
\begin{proof}
Suppose for a contradiction that we have a vertex $v \in G_0\backslash
A$ with at least $5/9$ths of its neighbors in $A$. Then $v$ also has
$5/9$ths of its neighbors in $V$, but $V$ is closed in $G_0$, so $v$
should be in $V$. However, in $G_0|V$, the vertex $v$ can only have lost
neighbors from outside $A$, so in $G_0|V$, $v$ also has at least
$5/9$ths of its neighbors in $A$. This contradicts that $A$ is closed
in $G_0|V$.
\end{proof}

We say that a vertex set $A\subseteq V_0$ {\em dominates\/} the
$\eps$-spectral cluster $\cluster$ if
$\vol_{G_0}(\cluster\cap V)>(1-3\,\eps)\vol_{G_0}(\cluster)$. We will
show that if $A$ is closed and has more than $2/3$ of the volume of $\cluster$,
then it dominates $\cluster$:
\begin{lemma}\LemmaName{closure-property}
If $A$ is closed in $G_0$ and $\vol_{G_0}(\cluster\cap A)\ge \frac 23 \vol_{G_0}(\cluster)$,
then $A$ dominates $\cluster$.
\end{lemma}
\begin{proof}
By assumption $r=\vol_{G_0}(\cluster\setminus A)/\vol_{G_0}(\cluster)\leq 1/3$. By the spectral expansion
of $\cluster$, we have 
\[E(\cluster\setminus A,\cluster\cap A)\geq (r(1-r)-\eps)\vol_{G_0}(\cluster)\geq (1-r-\eps/r)\vol_{G_0}(\cluster\setminus A).\]
The average degree from $\cluster\setminus A$ to $A$ is thus at least
$(1-r-\eps/r)$, which is bigger than $5/9$ if $3\eps\leq r\leq 1/3$
and $\eps\leq 1/27$ (we have $\eps<1/2000000$).  This would contradict
that $A$ is closed, so we conclude that $r<3\eps$, and hence that
$\vol_{G_0}(\cluster\cap A)/\vol_{G_0}(\cluster)=1-r> 1-3\eps$.
\end{proof}
We shall frequently use the fact that if $V$ dominates $\cluster$, then within $V\cap \cluster$,
there is not much difference between volume in $G$ and $G_0$. More precisely,
\begin{lemma}\LemmaName{approx-G0}
Consider an arbitrary set $S\subset V$. If $V$ dominates $\cluster$ then $\vol_{G_0}(\cluster\cap V)
\leq \vol_G(\cluster\cap S)+4\,\eps\vol_{G_0}(\cluster)$.
\end{lemma}
\begin{proof}
The edges from $\cluster\cap S$ in $G_0$ that are not in $G$, are either edges
leaving $\cluster$, of which, by isolation of $\cluster$, there are at most $\eps\vol_{G_0}(\cluster)$,
or edges to $\cluster\setminus V$, of which there at at most $\vol_{G_0}(\cluster\setminus V)=
\vol_{G_0}(\cluster)-\vol_{G_0}(V\cap \cluster)\leq 3\eps \vol_{G_0}(\cluster)$.
\end{proof}
We now continue describing our algorithm to isolate $\cluster$ (plus any other 
$\eps$-spectral cluster). As our invariant, we are given some induced subgraph
$G=(V,E)$ of $G_0$ such that 
\begin{enumerate}
\item[(1)] $V$ is closed in $G_0$.
\item[(2)] $V$ dominates $\cluster$.
\end{enumerate}
To maintain this invariant, we want to show that \clusteralg
finds a cut $(S,\widebar S)$ where both sides are closed and
one side dominates $\cluster$. We also want to make sure that both
sides are non-empty. By definition, this follows if the conductance
is below $1$, and we will, in fact, always keep the conductance below $1/9$.

The first step of \clusteralg is to use \Theorem{cheeger} to find a cut $(S,\widebar{S})$ of $G$ 
satisfying \Equation{cheeger}; say $|S| \le |\widebar{S}|$ (line 2 of Algorithm~\ref{alg:cluster}).
If $\phi(G,S)\geq 1/500$, we will just return $\{V\}$. The following
lemma states that $V$ does not have too much volume in $G$ outside $\cluster$. 
\begin{lemma}\LemmaName{smallextra} If $\phi(G,S)\geq 1/500$, then $\vol_G(V\setminus \cluster)\leq 1000000\,\eps\vol_{G_0}(\cluster)$.
\end{lemma}
\begin{proof} 
By isolation of the $\eps$-cluster $\cluster$ in $G_0$, we have
$|\partial_{G_0}(\cluster)|\leq \eps \vol_{G_0}(\cluster)$.  However,
$\partial_G(V\setminus \cluster)=\partial_G(V\cap
\cluster)\subseteq \partial_{G_0}(\cluster)$,
so $|\partial_G(V\setminus \cluster)|\leq \eps \vol_{G_0}(\cluster)$. 

Suppose $\vol_G(V\setminus \cluster)> 1000000\,\eps\vol_{G_0}(\cluster)$. We want to
conclude that $\phi(G,V\setminus \cluster)< 1/1000000$. To do that, we also
need to argue that $\vol_G(V\cap \cluster)>
1000000\,\eps\vol_{G_0}(\cluster)$. Since $V$ dominates $\cluster$, we have
$\vol_{G_0}(V\cap \cluster)> (1-3\,\eps)\vol_{G_0}(\cluster)$ and
by \Lemma{approx-G0}, we have
\[\vol_{G}(V\cap \cluster)\geq \vol_{G_0}(V\cap \cluster)-4\,\eps\vol_{G_0}(\cluster)
\leq (1-7\,\eps)\vol_{G_0}(\cluster)>1000000\,\eps\vol_{G_0}(\cluster).\]  
The last equality follows because $\eps \leq 1/2000000$.
We have now proved that $\vol_G(V\setminus \cluster)>
1000000\,\eps\vol_{G_0}(\cluster)$ implies that 
$\phi(G)\leq \phi(G,V\setminus \cluster)< 1/1000000$. By
\Equation{cheeger} this contradicts that we did not find a cut of size less
than $2\sqrt{\phi(G)}< 1/500$.
\end{proof}
The last cleaning in \main will reduce the volume of $V\setminus \cluster$ in $G_0$
as required for Theorem \ref{thm:cluster}. We shall return to that later.
Below we assume that we found a low conductance cut $(S,\widebar S)$ of $G$
with $\phi(G,S)\leq 1/500$. We are going to move vertices between $S$ and $\widebar S$, and will always maintain a conductance below $1/9$.

\paragraph{Local improvements towards closure}
We are now going to move vertices between $S$ and $\widebar S$ to make sure
that $\cluster$ is mostly contained in one side. As a first step, we will
make sure that both sides are closed (line 5 of Algorithm~\ref{alg:cluster}). This is done iteratively. If one
vertex has at least $5/9$ths fraction of its neighbors on the other
side of the cut, we move it to the other side, calling it a {\it
local improvement} (c.f. Algorithm \ref{alg:improve}). When no more local improvements are possible, both
sides $S$ and $\widebar S$ must be closed in $G$. We call these moves improving
because they always improve both cut size and conductance, as described more
formally below.
\begin{lemma}\LemmaName{improve}
Consider a local improvement moving a vertex $v$. It reduces
the cut size by at least $d_G(v)/9$. It also improves
the conductance if it was below $1/9$ before the move.
\end{lemma}
\begin{proof}
When we move $v$, we replace at least $\frac 59d_G(v)$ cut edges by
at most $\frac 49d_G(v)$ cut edges, reducing the cut
size by at least $d_G(v)/9$. The volume moved is $d_G(v)$, so if
$\phi(G,S)<1/9$ and the new cut is $(S',\widebar S')$, we get
\[\phi(G,S')=\frac{|\partial_G (S')|}{\min\{\vol_G(S'),\vol_G(\widebar S')\}}
\leq \frac{|\partial_G (S)|-d_G(v)/9}{\min\{\vol_G(S),\vol_G(\widebar S)\}-
d_G(v)}<\phi(G,S).\]
\end{proof}

When both sides are closed under local improvements, we have a situation
where either $\cluster$ is almost completely dominated by one side, or it roughly
balanced between the two sides. More precisely,
\begin{lemma}\LemmaName{extreme-or-mid}
When both $S$ and $\widebar S$ are closed in $G$ (and hence in $G_0$), 
then either $\vol_{G_0}(S\cap \cluster)<3\,\eps\vol_{G_0}(\cluster)$, or
$(1/3-3\,\eps)\vol_{G_0}(\cluster)< \vol_{G_0}(S\cap \cluster)< (2/3)\vol_{G_0}(\cluster)$,
or $\vol_{G_0}(S\cap \cluster)>(1-3\,\eps)\vol_{G_0}(\cluster)$.
\end{lemma}
\begin{proof}
This follows almost directly from \Lemma{closure-property}. We just
note that since $\vol_{G_0}(S\cap \cluster)+\vol_{G_0}(\widebar S\cap \cluster)=
\vol_{G_0}(V\cap \cluster)\geq (1-3\,\eps)\vol_{G_0}(\cluster)$,  if $\vol_{G_0}(S)\leq
(1/3-3\,\eps)\vol_{G_0}(\cluster)$, then $\vol_{G_0}(\widebar S\cap \cluster)\geq 2/3\vol_{G_0}(\cluster)$,
and then, by \Lemma{closure-property}, we have $\vol_{G_0}(\widebar S\cap \cluster)>
(1-3\,\eps)\vol_{G_0}(\cluster)$, and hence $\vol_{G_0}(S\cap \cluster)<3\,\eps\vol_{G_0}(\cluster)$.
\end{proof}

\paragraph{Grabbing for dominance}
Having made sure that both sides are closed, as described
in \Lemma{extreme-or-mid}, we now have $\cluster$ either almost
completely dominated by one side, or roughly balanced between both
sides.  We will now introduce a {\em grab\/} operation
(c.f. Algorithm \ref{alg:grab}) that in the balanced case will help
$S$ get dominance.  The grab operation itself is simple. In one round,
it moves every vertex to $S$ that before the grab has more than
$1/6$th of its neighbors in $S$.  An important point here is that
contrary to the local improvements, the grabbing is not recursive. 
We are actually going to do this grabbing twice, interspersed with 
local improvements (c.f., Algorithm \ref{alg:cluster}): First we
grab from $S$, then we do local improvements of $S$, only moving
vertices to $S$, thus only closing $S$. We do this grabbing followed by local
improvements from $S$ twice, calling it the {\em big expansion\/} of $S$ (Algorithm \ref{alg:cluster} lines 6--9). 
Finally we do local improvements of $\widebar S$. 
We will prove that when all this is completed, then
one of the two sides dominate $\cluster$, but first we have to argue that
the conductance stays low.
\begin{lemma}\LemmaName{conductance} With the above grabs and local improvements, the conductance always stays below $1/9$.
\end{lemma}
\begin{proof} Our starting point is a Cheeger cut $(S,\widebar S)$ with conductance
$\phi(G,S)\leq 1/500$. From \Lemma{improve} will never increase the
conductance if it is below $1/9$. Thus, it is only the grabs that can
increase the conductance. Let $S'$ be the result of the grab from $S$.
We know that all the vertices grabbed from $\widebar S$
have at least $1/6$th of their neighbors in $S$. This means that
$|\partial_G(S')|\leq 5 |\partial_G(S)|$ and
$\vol_G(S')\leq\vol_G(S)+6\partial_G(S)$.  It follows that
\begin{align*}
\phi(G,S')\leq &\frac{|\partial_G(S')|}{\min\{\vol_G(S'),\vol_G(\widebar{S'})\}}\leq
\frac{5|\partial_G(S)|}{\min\{\vol_G(S),\vol_G(\widebar S)\}-6|\partial_G(S)|}\\
&\leq \frac{5|\partial_G(S)|}{(1-6\phi(G,S))\min\{\vol_G(S),\vol_G(\widebar S)\}}
=\frac{5}{(1-6\phi(G,S))}\phi(G,S).
\end{align*}
Before the first grab, we have conductance below $1/500$, which
the grab now increases to at most $6/((1-6/500)500)<1/80$.
The second grab can then increase the conductance to at most
$6/((1-6/80)80)<1/12$, so, in fact, the conductance will always stay below
$1/12$.
\end{proof}

Below we first argue that if we started with balance, then after 
the big expansion of $S$, we have $\vol_{G_0}(S\cap \cluster)\geq (1-3\,\eps)
\vol_{G_0}(\cluster)$.

\begin{lemma}\LemmaName{grab}
If $S$ is closed and $\vol_{G_0}(\cluster\setminus S)\leq
(2/3)\vol_{G_0}(\cluster)$, then $S$ dominates $\cluster$ after the big expansion from $S$.
\end{lemma}
\begin{proof}
Let $r=\vol_{G_0}(\cluster\setminus S)/\vol_{G_0}(\cluster)\leq 2/3$. 
If we can show that the big expansion of $S$ brings $r\leq 1/3-3\,\eps$, then
we are done by \Lemma{closure-property} since $S$ is closed and
$\vol_{G_0}(S\cap \cluster)\geq (1-r-3\,\eps)\vol_{G_0}(\cluster\setminus S)$.

We study the effect of a single grab from a closed set $S$. 
By spectral expansion, we have $|E(\cluster\setminus S,S\cap \cluster)|
\geq (r(1-r)-\eps)\vol_{G_0}(\cluster)$.

Let $T$ be the set of vertices from $\cluster$ moved by the grab, that is,
$T$ is the set of vertices from $\cluster\setminus S$ with at least $1/6$th
of their neigbors in $S$. Also, let $B$ be the set of vertices from
$\cluster\setminus S$ that are not moved, that is, $B=\cluster\setminus (S\cup
T)$. Then
\[|E(S\cap \cluster,B)|\leq |E(S,B)|<\vol_G(B)/6\leq \vol_{G_0}(\cluster\setminus S)/6=
r\vol_{G_0}(\cluster)/6.\] The number of edges from $S\cap \cluster$ to vertices
from $\cluster\setminus S$ that will be grabbed is therefore at least
\[|E(S\cap \cluster,T)|=|E(S\cap \cluster,\cluster\setminus S)|-E(S\cap \cluster,B)
\geq (r(1-r)-\,\eps-r/6)\vol_{G_0}(\cluster)=(5r/6-r^2-\,\eps)\vol_{G_0}(\cluster).\]
The function $f(r)=5r/6-r^2$ is convex, having its minimum over any
interval in its end points. We are only concerned with $r\in
[(1/3-3\,\eps),2/3]$, where the smallest value is $f(2/3)=1/9$. The
number of edges to grabbed vertices is thus $|E(S\cap \cluster,T)|\geq 
(1/9-\,\eps)\vol_{G_0}(\cluster)$.

We also know that before a grab, the set $S$ is closed. Therefore, 
every vertex $v$ outside $S$ has less than $\frac 59 d_G(v)$ edges 
going to $S$, and no more going to $S\cap \cluster$. We conclude
that 
\[\vol_{G_0}(T)\geq \vol_{G}(T)\geq \frac 95|E(S\cap \cluster,T)|\geq 
\frac 95(1/9-\,\eps)\vol_{G_0}(\cluster)\geq  (1/5-2\,\eps)\vol_{G_0}(\cluster).\]
Each grab thus decreases $r=\vol_{G_0}(\cluster\setminus S)/\vol_{G_0}(\cluster)$ by 
$(1/5-2\,\eps)$. Starting from $r\leq 2/3$, with two grabs followed by
local improvements of $S$, we thus end up with $
r\leq 2/3-2(1/5-2\,\eps)\leq 4/15+4\,\eps$. This is less than the
desired $1/3-3\,\eps$ if $\eps<1/105$, which is indeed the case.
\end{proof}

Next, we will show that $\widebar S$ dominated $\cluster$ before the big expansion of $S$,
then $\widebar S$ will also dominate $\cluster$ at the end.
\begin{lemma}\LemmaName{grab-less}
Suppose that $S$ and $\widebar S$ are closed and that $\widebar S$ dominate $\cluster$.
Then $\widebar S$ will also dominate $\cluster$ after
a big expansion from $S$ followed by local improvements of $\widebar S$.
\end{lemma}
\begin{proof}
We want to show that after the big expansion of $S$, we still have 
$\vol_{G_0}(\cluster\cap \widebar S)\geq\frac23 \vol_{G_0}(\cluster)$. This can only increase
when we subsequently do local improvements of $\widebar S$, closing $\widebar S$,
and then the result follows from \Lemma{closure-property}.

Stepping back to $S$ before the big expansion, we are trying to bound
how much volume from $\vol_{G}(\cluster\setminus S)$ that can be moved
to $S$. Both for local improvements and for grabbing, the chance
of moving $v$ to $S$ increases the more neighbors $v$ has in $S$.
Thus, we maximize the potential of moving vertices from $\cluster$ to $S$ if
we assume that we start with all of $V\setminus \cluster$ in $S$. The only
cut edges are now those between $S$ and $\cluster\setminus S$.

When a local improvement moves a vertex $v$ to $S$, we replace 
at least $\frac 59 d_G(v)$ cut edges with at most $\frac 49d_G(v)$ cut edges
while increasing $\vol_G(\cluster\cap S)$ by $d_G(v)$. It follows that if 
local improvements increase $\vol_G(S)$ by $x$, then they
decrease $|\partial_G(S)|$ by at least $x/9$.

When we grab a set $T$ of vertices to $S$, we know that
each vertex $v\in T$ had at least $\frac 16 d_G(v)$ of its edges to $S$.
When the grab is done, setting $S'=S\cup T$, we have
at most $\frac 56 d_G(v)$ edges from $v$ to $\cluster\setminus S'$. We
conclude that $|\partial_G(S')| \le 5|\partial_G(S)|$ 
and that $\vol_G(S') \le \vol_G(S) + 6 |\partial_G(S)|$.
Subsequent local improvements can further increase the volume by
at most $9|\partial_G(S')|\leq  45|\partial_G(S)|$. Thus, when
we from $S$ do one grab followed by local improvements, we
increase the volume $\vol_G(S)$ by at most 
$51|\partial_G(S)|$.

We now note that if we did local improvements of $S$ before the
grab, then we would only get a smaller bound. More
precisely, if the the local improvements moved volume $x$ to $\vol_G(S)$,
then $|\partial_G(S)|$ would be reduced by $x/9$ before the grab, and then 
the total volume increase would be at most 
$x+51(|\partial_G(S)|-x/9)\leq 51|\partial_G(S)|$.

We now study the big expansion from $S$, starting with a grab 
followed by any number of local improvements followed by another grab followed 
any number of local improvements. From the
above analysis, it follows that the total increase in 
$\vol_G(S)$ is at most 
$6|\partial_G(S)|+51\cdot 5|\partial_G(S)| \leq 261 |\partial_G(S)|$. This
bounds the total volume moved from $\vol_G(\cluster\setminus S)$ to $\vol_G(S)$.

We started with $\vol_G(\cluster\cap S)\leq 3\,\eps\vol_{G_0}(\cluster)$. Also,
originally, before the big expansion, the set $\widebar S$ was closed, so
vertices from $S$ had less than a fraction $5/9$ of their neighbors in
$\widebar S$. This fraction could only be reduced when we artificially
added all of $V\setminus \cluster$ to $S$. After this artificial change,
which only improved movement of vertices from $\cluster\setminus S$ to $S$,
we had $|\partial S|\leq \frac 59\vol_G(\cluster\cap S)+|\partial_G (\cluster)|\leq
(5/3+1)\eps\vol_{G_0}(\cluster)$. Thus, after the big expansion, we end up with
\[\vol_G(\cluster \cap S)\leq 
 (3+261(5/3+1)\eps\vol_{G_0}(\cluster)<699\;\eps \vol_{G_0}(\cluster).\]
Finally, by \Lemma{approx-G0}, we have 
$\vol_{G_0}(\cluster\cap S)\leq \vol_{G}(\cluster\cap S)+4\,\eps\vol_{G_0}(\cluster)\leq 
703\;\eps \vol_{G_0}(\cluster)$,
which with $\eps\leq 1/1000000$, is much less than the required $(1/3-3\,\eps)
\vol_{G_0}(\cluster)$.
\end{proof}
Finally we need
\begin{lemma}\LemmaName{local-preservation}
Suppose that $S$ is closed and dominates $\cluster$. Then both
of these properties are preserved if we do local improvements of $\widebar S$.
\end{lemma}
\begin{proof}
It sufficies to consider a local improvement moving a single vertex
$v$ to $\widebar S$.

First we show that $S$ being closed is preserved when we do a local
improvement from $\widebar S$ moving a vertex $v$ to $\widebar S$. Suppose for
contradiction that we afterwards had some vertex $u\in \widebar S$ with
$5/9$ths of its neighbors in $S$. Then we cannot have $u=v$; for we
moved $v$ because it had $5/9$ths of its neighbors in $\widebar
S$. However, any other $u\neq v$ can only have fewer neighbors in
$S$ after the move.

Since $S$ starts dominating, we start with $\vol_{G_0}(\cluster\setminus
S)\leq 3\,\eps\vol_{G_0}(\cluster)$.  The dominance can only be
affected if $v\in \cluster$, and we only move $v$ to $\widebar S$ because it has
at least $5/9$ths of its neighbors in $\widebar S$. Trivially $v$ can have
at most $\vol_{G_0}(\cluster\cap\widebar S)+|\partial_{G_0}(\cluster)|\leq 4\,\eps \vol_{G_0}(\cluster)$
neighbors in $\widebar S$, so we conclude that $d_G(v)\leq \frac 95\cdot 
4\,\eps \vol_{G_0}(\cluster)<8\,\eps \vol_{G_0}(\cluster)$, and by \Lemma{approx-G0},
we have $d_{G_0}(v)\leq d_G(v)+4\,\eps \vol_{G_0}(\cluster)<12\,\eps \vol_{G_0}(\cluster)$.
Thus we end up with $\vol_{G_0}(\cluster\setminus
S)<15\,\eps\vol_{G_0}(\cluster)$. This is much less that $(1/3-3\,\eps)\vol_{G_0}(\cluster)$.
Since $S$ remains closed, we conclude from \Lemma{closure-property} that
$S$ still dominates $\cluster$, hence that $\vol_{G_0}(\cluster\setminus
S)\leq 3\,\eps\vol_{G_0}(\cluster)$.
\end{proof}
Let us now sum up what we have proved. Suppose we found an initial cut 
$(S,\widebar S)$ of conductance below $1/500$. This is when we start
modifying the cut, and we need to show that our invariants are preserved for the recursive calls.

By \Lemma{local-preservation}, the local improvements of $S$
followed by the final local improvements of $\widebar S$ imply that both
sets end up closed, so invariant (1) is satisfied.  The hard part was
ensure that one side dominates $\cluster$. We had 
\Lemma{conductance} showing that we keep the conductance
below $1/9$. First we did local improvements of both sides. By \Lemma{extreme-or-mid}, we end up either with one side dominating $\cluster$, or with $\cluster$ roughly balanced
between the sides. If $\cluster$ is balanced or dominated by $S$,
then after by big expansion from $S$, by \Lemma{grab}, $S$ dominates $\cluster$. 
By \Lemma{local-preservation}, this dominance is preserved when
we do the final local improvements of $\widebar S$. If instead $\widebar S$ dominated
$\cluster$, then, by \Lemma{grab-less} this is dominance is preserved. Thus we always
end up with one side dominating $\cluster$, so invariant (2) is satisfied. We conclude
that $\cluster$ is always dominated by one branch of the recursion. 

The recursion finishes when no small cut is found. We have $V$ dominating $S$,
and by \Lemma{smallextra}, we have 
$\vol_G(V\setminus \cluster)\leq 1000000\,\eps\vol_{G_0}(\cluster)$. Thus we have
proved
\begin{theorem}\label{thm:cluster-rec} For any given $\eps\leq 1/2000000$ and graph $G_0$, in polynomial time, \clusteralg (Algorithm \ref{alg:cluster}) finds
a partitioning $\{U_1,\ldots U_\ell\}$ of the vertex set so that every {\em
$\eps$-spectral cluster\/} $\cluster$ of $G_0$ matches some set $U_i$ in the
sense that
\begin{itemize}
\item $\vol(\cluster\setminus U_i)\leq 3\,\eps\vol(\cluster)$.
\item $\vol_{G_0|U_i}(U_i\setminus \cluster)\leq 1000000\,\eps\vol(\cluster)$
\end{itemize}
\end{theorem}
The only detail missing in proving Theorem \ref{thm:cluster} is that
we want a bound on $\vol_{G_0}(U_i\setminus \cluster)$ rather than $\vol_{G_0|U_i}(U_i\setminus \cluster)$.
This
is where the last cleaning of \main (Algorithm \ref{alg:main}) comes in.
It removes all vertices from $U_i$ with more than $5/9$ths of their
neighbors outside $U_i$, so now, for all $v\in U_i$, we have $d_{G_0}(v)\leq
\frac 94 d_{G_0|U_i}(v)$.
\begin{lemma}\LemmaName{cleaning} The cleaning of $U_i$ preserves that $U_i$ dominates $\cluster$,
and now $\vol_{G_0}(U_i\setminus \cluster)\leq \frac 94 \vol_{G_0|U_i}
(U_i\setminus \cluster)\leq 225000000\,\eps\vol(\cluster)$.
\end{lemma}
\begin{proof}
Recall that $U_i$ before the cleaning is closed and dominates $\cluster$. The
cleaning has exactly the same effect as if we did local improvements of
$V(G_0)\setminus U_i$, so by \Lemma{local-preservation}, we preserve
both that $U_i$ is closed, and that it dominates $\cluster$.
\end{proof}
Theorem \ref{thm:cluster-rec} and \Lemma{cleaning} immediately imply the
statement of Theorem \ref{thm:cluster}.

%% file: appendix.tex
\section*{Appendix}

\section{Heavy hitters reduction from small to large $\eps$}\SectionName{reduction}

The main theorem of this section is the following.

\begin{theorem}\TheoremName{reduce-threshold}
Let $0 < \eps, \delta < 1/2$ be such that $1/\eps^2 > \log (1/\delta)$. Then there is a reduction from turnstile $\ell_2$ heavy hitters with parameter $\eps$ and failure probability $\delta$ to $z = \ceil{1/(\eps^2\log (1/\delta))}$ separate $(1/\sqrt{t})$-heavy hitters problems for $t = C\log(1/\delta)$ for some constant $C > 0$, and where each such problem must be solved with failure probability at most $\delta' = \delta/(3z)$. Specifically, given an algorithm for the $(1/\sqrt{t})$-heavy hitters problem with these parameters using space $S$, update time $t_u$ and query time $t_q$, the resulting algorithm for $\eps$-heavy hitters from this reduction uses space $O(z\cdot S + \log (1/\delta))$, and has update time $t_u + O(\log (1/\delta))$ and query time $O(z\cdot t_q)$.
\end{theorem}
\begin{proof}
We pick a hash function $h:[n]\rightarrow[z]$ at random from a $\Theta(\log(1/\delta))$-wise independent family as in \cite{CarterW79}. Such $h$ requires $O(\log (1/\delta))$ space to store and can be evaluated in $O(\log (1/\delta))$ time. We also instantiate $z$ independent $(1/\sqrt{t})$-heavy hitter data structures $P_1,\ldots, P_z$ each with failure probability $\delta'$. Upon receiving an update $(i,v)$ in the stream, we feed the update $(i,v)$ to $P_{h(i)}$. To answer a query, we first query each $P_j$ to obtain a list $L_j$, then we output $L = \cup_{j=1}^z L_j$.

This concludes the description of the data structures stored, and the implementations of update and query. The space bound and running times thus follow.

We now argue correctness. Let $H$ be the set of $\eps$-heavy hitters of $x$, so that $|H| \le 2/\eps^2$. It follows by the Chernoff bound and union bound over all $j\in[z]$ that, for some constant $C>0$,
\begin{equation}
\Pr_h(\exists j\in [z],\ |H \cap h^{-1}(j)| > C \log (1/\delta)) < \delta/3 . \EquationName{few-heavy}
\end{equation}
We next invoke Bernstein's inequality, which says there exists constant $c>0$ such that for independent $X_1,\ldots,X_r$ each bounded by $K$ in magnitude with $\sigma^2 = \sum_i \E(X_i - \E X_i)^2$, 
$$
\forall\lambda>0,\ \Pr(|\sum_i X_i - \E\sum_i X_i| > \lambda) \lesssim e^{-c\lambda^2/\sigma^2} + e^{-c\lambda/K} .
$$
We will consider the collection of random variables $X_i$ indexed by $i\in[n]\backslash H$ defined as follows (so $r = n - |H|$). Fix $j\in[z]$. Define $X_i = \mathbf{1}_{h(i)=j}\cdot x_i^2$. It then follows that $K < \eps^2\cdot\|x_{\proj{1/\eps^2}}\|_2^2$ and
$$
\sigma^2 = \sum_{i\notin H} x_i^4 (1/z - 1/z^2) \le \frac{\eps^2 }z\cdot \|x_{\proj{1/\eps^2}}\|_2^2\cdot\sum_{i\notin H} x_i^2 = \frac{\eps^2 }z\cdot \|x_{\proj{1/\eps^2}}\|_2^4
$$
It then follows by Bernstein's inequality and a union bound over all $j\in[z]$ that
\begin{equation}
\Pr_h(\exists j\in [z],\ \sum_{\substack{i\notin H\\h(i) =j}} x_i^2 > C\eps^2\log (1/\delta)\|x_{\proj{1/\eps^2}}\|_2^2) < \delta/3 .\EquationName{light-weight}
\end{equation}

Now, the hash function $h$ divides the input stream into $z$ separate substreams, with $x(j)$ denoting the vector updated by the substream containing $i$ with $h(i) = j$ (i.e.\ $(x(j))_i = \mathbf{1}_{h(i) = j}\cdot x_i$). Conditioned on the events of \Eqsub{few-heavy}, \Eqsub{light-weight} occurring, we have that any $i\in H$ satisfies $x_i^2\ge \eps^2\|x_{\proj{1/\eps^2}}\|_2^2 \ge 1/(C\log (1/\delta))\cdot \|x(h(i))_{\proj{C\log (1/\delta)}}\|_2^2$. Thus if we also condition on every $P_j$ succeeding, which happens with probability at least $1-z\delta' = 1-\delta/3$ by a union bound, then our output $L$ is correct. Thus our overall failure probability is at most $\delta$.
\end{proof}

\section{Partition heavy hitters and the $b$-tree}\SectionName{partition}

We first define a generalization of the heavy hitters problem that we call {\em partition heavy hitters}. The standard heavy hitters problem is simply the special case when $\mathcal{P} = \{ \{1\}, \ldots, \{n\} \}$, and the oracle $\oracle$ is the identity map.

\begin{definition}
In the {\em $\ell_2$ partition heavy hitters problem} there is some error parameter $\eps\in(0,1/2)$. There is also some partition $\mathcal{P} = \{S_1,\ldots,S_N\}$ of $[n]$, and it is presented in the form of an oracle $\oracle:[n]\rightarrow[N]$ such that for any $i\in[n]$, $\oracle(i)$ gives the index $j\in[N]$ such that $i\in S_j$. Define a vector $y\in\R^N$ such that for each $j\in[N]$,
$$
y_j = \sqrt{\sum_{i\in S_j} x_i^2} .
$$
The goal then is to solve the $\ell_2$ $\eps$-heavy hitters problem on $y$ subject to streaming updates to $x$. That is, $x$ is being updated in a stream, and we should at the end output a list $L\subset[N]$, $|L| = O(1/\eps^2)$, such that $L$ contains all the $\eps$-heavy hitters of $y$.

A related problem is the {\em $\ell_2$ partition point query problem}, in which queries takes as input some $j\in[N]$ and must output a value $\tilde{y}_j$ such that $\alpha_1 y_j - \eps\|y_{\proj{1/\eps^2}}\|_2 \le \tilde{y}_j \le \alpha_2 y_j + \eps\|y_{\proj{1/\eps^2}}\|_2$ for some approximation parameters $\alpha_1, \alpha_2, \eps$. We call such a solution an $(\alpha_1, \alpha_2, \eps)$-partition point query algorithm.
\end{definition}

We now describe the $b$-tree, introduced in \cite{CormodeH08} (and also called the \hcs there and in other parts of the literature) and analyze its performance for solving the partition heavy hitters problem. The $b$-tree was only suggested in \cite{CormodeH08} for the strict turnstile model, and a detailed analysis was not given. Here we give a fully detailed analysis showing that, in fact, that structure gives good guarantees even in the general turnstile model, and even for the partition heavy hitters problem. Before delving into the $b$-tree, we will analyze the performance of the \pcs, a very slight modification of the \cs, for the partition point query problem.

\subsection{\pcs}

In the idealized \pcs one chooses independent hash functions $h_1,\ldots,h_L:[N]\rightarrow[k]$ and $\sigma_1,\ldots,\sigma_L:[N]\rightarrow\{-1,1\}$ from pairwise independent families, and $\beta_1,\ldots,\beta_L:[n]\rightarrow\R$ independently from a $C_\alpha$-wise independent family for some small constant $\alpha>0$. The marginal distribution of $\beta_t(i)$ for any $i\in[n]$ and $t\in[L]$ is the standard gaussian $\mathcal{N}(0,1)$, and $C_\alpha$ is chosen so that for any $z\in\R^n$, $\sum_i z_i \beta_t(i)$ has Kolmogorov distance at most $\alpha$ from $\mathcal{N}(0,1)$ (see \Theorem{knw}).
One also initializes counters $C_{a,b}$ to $0$ for all $(a,b) \in [L]\times [k]$, each consuming one machine word. The space is then clearly $O(kL)$ machine words.

The update algorithm for $x_i\leftarrow x_i + \Delta$ performs the change $C_{t,h_t(\oracle(i))}\leftarrow C_{t,h_t(\oracle(i))} + \Delta\cdot \sigma_t(\oracle(i))\cdot \beta_t(i)$ for each $t\in[L]$. To answer an $\ell_2$ point query to index $i\in[N]$, one outputs $\tilde{y}_i$ being the median value of $|C_{t,h_t(i)}|$ across $t$. In the actual, non-idealized version of \pcs, the marginal distribution of each $\beta_t(i)$ is a discretized gaussian with precision fitting in a single machine word; see \cite[Section A.6]{KaneNW10} for details. We henceforth just discuss the idealized version. In fact, for the types of instances of partition heavy hitters that our heavy hitters algorithm needs to solve, the $\beta_t$ can be replaced with pairwise independent hash functions mapping to $\{-1,1\}$ (see \Remark{no-gaussian}), which would make an implementation less complicated. We instead have chosen to analyze the version with gaussians since doing so then does not require us to limit the types of partition heavy hitter instances we can handle, thus providing a more general guarantee that may prove useful as a subroutine in future works.

Before analyzing \pcs, we state a useful theorem from \cite{KaneNW10}.

\begin{theorem}{{\cite{KaneNW10}}}\TheoremName{knw}
Let $z\in\R^n$ and $0<\alpha<1$ be arbitrary. There exists a constant $C_\alpha$ depending on $\alpha$ such that if $\beta_1,\ldots,\beta_n$ are $C_\alpha$-wise independent standard gaussians and $g$ is standard gaussian, then
$$
\forall A < B\in\R,\ |\Pr(\sum_i \beta_i z_i \in[A, B]) - \Pr(g \in [A, B])| < \alpha .
$$
\end{theorem}
In fact $C_\alpha$ can be taken as some value in $O(1/\alpha^2)$ by combining the FT-mollification construction in \cite{DiakonikolasKN10} with Lemma 2.2 and the argument in the proof of Theorem 2.1 in \cite{KaneNW10}. Using the FT-mollification construction in \cite{KaneNW10} would imply the weaker statement that $C_\alpha = \Omega(\log^6(1/\alpha)/\alpha^2)$ suffices. In our current setting either is acceptable since $\alpha$ is some fixed constant.

The following lemma shows the correctness of \pcs, following an argument similar to that of \cite[Lemma 3]{CharikarCF04}. 

\begin{lemma}\LemmaName{pcs}
For $L \gtrsim \log(1/\delta)$ and $k\gtrsim 1/\eps^2$, 
\begin{equation}
\forall j\in[N],\ \Pr(\tilde{y}_j \notin [(1/20) y_j - \eps\|y_\projeps\|_2, 3 y_j + \eps\|y_\projeps\|_2]) < \delta . \EquationName{pcs-condition}
\end{equation}
That is, the \pcs with these parameters is an $(\alpha_1, \alpha_2, \eps)$-partition point query structure with failure probability $\delta$, for $\alpha_1 = 1/20, \alpha_2 = 3$.
\end{lemma}
\begin{proof}
For $t\in[L]$, define $w_t\in\R^N$ with $(w_t)_j = \sum_{i\in S_j} \beta_t(i) x_i$. Then by \Theorem{knw}, $\Pr((w_t)_j \in [-\alpha_1,\alpha_2]) \le \Pr(y_jg\in [-\alpha_1,\alpha_1]) + \alpha_1$ for $g$ a standard gaussian. Since $\Pr(|g| \le \epsilon) \le \epsilon \sqrt{2/\pi}$ for any $\epsilon>0$, 
\begin{equation}
\Pr(|(w_t)_j| > \alpha_1 y_j) \ge 1 - \alpha_1(1 + \sqrt{2/\pi}) > 8/9 . \EquationName{anticonc}
\end{equation}
Also by pairwise independence of $\beta_t$, $\E (w_t)_j^2 = y_j^2$, so
\begin{equation}
\Pr(|(w_t)_j| > \alpha_2y_j) < 1/\alpha_2^2 = 1/9 . \EquationName{small-w}
\end{equation}
Next, let $H\subset[N]$ be the set of indices of top $1/(4\eps^2)$ entries of $y$ in magnitude. Note
$$
\E_{h_t,\beta_t} |\sum_{\substack{h_t(\oracle(i)) = h_t(j)\\\oracle(i)\notin H\cup\{j\}}} \beta_t(i) x_i| \le \frac 1{\sqrt{k}}\cdot \|y_\projeps\|_2
$$
by pairwise independence of $\beta_t, h_t$. Thus
\begin{equation}
\Pr_{h_t,\beta_t}(|\sum_{\substack{h_t(\oracle(i)) = h_t(j)\\\oracle(i)\notin H\cup\{j\}}} \beta_t(i) x_i| > \eps\|y_\projeps\|_2) < 1/9 \EquationName{small-noise}
\end{equation}
by Markov's inequality. Also $\E_{h_t} |h_t^{-1}(h_t(j))\cap (H\backslash \{j\})| \le |H|/k < 1/9$, so
\begin{equation}
\Pr(\exists j'\in (H\backslash \{j\}) : h_t(j') = h_t(j)) < 1/9\EquationName{isolate-j}
\end{equation}
by Markov's inequality. Thus by a union bound, the conditions of \Eqsub{anticonc}, \Eqsub{small-w}, \Eqsub{small-noise}, and \Eqsub{isolate-j} all happen simultaneously with probability at least $5/9$. Now note when all these events occur, 
\begin{equation}
|C_{t, h_t(j)}| \in [\alpha_1 y_j - \eps\|y_\projeps\|_2, \alpha_2 y_j + \eps\|y_\projeps\|_2] \EquationName{good-counter}
\end{equation}
Thus by a chernoff bound, the probability that the median over $t\in[L]$ of $|C_{t,h_t(j)}|$ does not satisfy \Eqsub{good-counter} is $\exp(-\Omega(l)) < \delta$.
\end{proof}

\begin{remark}\RemarkName{no-gaussian}
\textup{
It is possible to obtain the main result of this paper, i.e.\ our bounds for (non-partition) heavy hitters, without ever making use of \Theorem{knw}. This might be useful for pedagogical reasons, since the proof of \Theorem{knw} is somewhat complicated. We sketch how to do so here. We chose to present the approach above since it divides our final algorithm into a more abstracted set of subproblems, which we felt made it easier to describe. Furthermore, the ``partition heavy hitters'' problem as formulated above seems quite clean and may find uses as a subroutine elsewhere.
}

\textup{
The main idea to avoid \Theorem{knw} is as follows. First, in our application of partition heavy hitters to solve the vanilla heavy hitters problem (see the proof of \Theorem{clustering-reduction}), we are only interested in the case where each partition $S_j$ in $\mathcal{P}$ contains at most one $\eps$-heavy hitter $i\in[n]$ with respect to $x$. Furthermore, the other items in the same partition will, with good probability over the choice of the hash function $h_j$ used in the \es, have much smaller total $\ell_2$ mass than $|x_i|$. These observations can be used to show that, for applying partition heavy hitters to heavy hitters as in \Theorem{clustering-reduction}, the $\beta_t$ above can be replaced with pairwise independent hash functions mapping to $\{-1,1\}$. This ensures that the ``noise'', i.e.\ the non-heavy indices in the same partition as $i$, with say $8/9$ probability do not subtract much away from $|x_i|$ so that $|(w_t)_j|$ will be $\Omega(|x_i|)$. Note that in the general partition heavy hitters problem the range of the $\beta_t$ cannot just be $\{-1,1\}$. The reason is that we want \Eqsub{anticonc}, \Eqsub{small-w}, \Eqsub{small-noise}, and \Eqsub{isolate-j} to hold simultaneously with some probability strictly larger than $1/2$, to allow invoking Chernoff. However, \Eqsub{anticonc} alone can fail with probability at least $1/2$ even when using fully independent $\beta_t:[n]\rightarrow\{-1,1\}$. For example, consider when a partition $j\in[N]$ contains exactly two elements with equal weight. Then with probability $1/2$ their signs are opposite, in which case $(w_t)_j$ equals zero. Note this is not an issue when $S_j$ contains a single heavy hitter that is much heavier than the total $\ell_2$ mass in the rest of the partition.
}

\textup{
The second issue is that \pcs is again used in the $b$-tree (see \Section{btree} below), which we use in our final heavy hitters algorithm. Unfortunately there, we do {\em not} have the guarantee that each partition contains one item much heavier than the combined mass of the rest. One way around this is to simply not use the $b$-tree at all, but to just use the \pcs itself. Doing so, though, would make our final query time blow up to $O(\eps^{-2}\log^C n)$ for a large constant $C$. This can be fixed, and one can achieve our same current query time, by instead implementing the data structures $P_j$ in \Section{turnstile} as themselves recursive instantiations of the \es data structure! Then after a constant number of levels of recursion (even one level down), one then implements the $P_j$ using the \pcs (see \Remark{recurse}). All in all, \cite{KaneNW10} is thus avoided. We do not delve further into the details since the bounds do not improve, and we feel that a fully detailed exposition is not worth the effort.
}
\end{remark}

We also record a lemma here which will be useful later.

\begin{lemma}\LemmaName{take-top}
Consider an instance of partition heavy hitters with partition $\mathcal{P} = \{S_1,\ldots,S_N\}$ and oracle $\oracle$. Suppose $A\subset[N]$, and $D$ is an $(\alpha_1, \alpha_2, \eps')$-partition point query structure for some $0 < \alpha_1 \le 1 \le \alpha_2$ which succeeds on point querying every $j\in A$, where $\eps' = (\alpha_1/3)\eps$. Then if $L\subseteq A$ is defined to be the $2/\eps''^2$ indices of $A$ with the largest point query results from $D$ for $\eps'' = \eps'/\alpha_2$, then $L$ contains all partition $\eps$-heavy hitters contained in $A$.
\end{lemma}
\begin{proof}
Let $\tilde{y}_j$ be the result of a point query of $j\in A$ using $D$. If $j\in A$ is a partition $\eps$-heavy hitter then
$$
\tilde{y}_j \ge \alpha_1\eps \|y_\projeps\|_2 - (\alpha_1/3)\eps\|y_{\projepstwo{\eps'}}\|_2 \ge (2/3)\alpha_1\eps\|y_{\projepstwo{\eps}}\|_2 .
$$
by \Equation{pcs-condition}. Meanwhile, if some $j'$ is not even an $\eps''$-partition heavy hitter, then
$$
\tilde{y}_{j'} < \alpha_2\eps''\|y_{\projepstwo{\eps'}}\|_2 + \eps'\|y_{\projepstwo{\eps'}}\|_2\le (2/3)\alpha_1\eps\|y_{\projepstwo{\eps'}}\|_2 .
$$
Thus every set which is not even a partition $\eps''$-heavy hitter has a point query value strictly less than every partition $\eps$-heavy hitter in $A$. Since there are at most $2/\eps''^2$ sets which are partition $\eps''$-heavy hitters, correctness follows, since $L$ by definition contains the $2/\eps''^2$ sets with largest point query values.
\end{proof}

\subsection{$b$-tree}\SectionName{btree}

Here we describe the $b$-tree for partition heavy hitters, which is a slight modification of the \hcs from \cite{CormodeH08} for non-partition strict turnstile heavy hitters. We assume $2\le b\le N$ is an integer power of $2$.

For what follows, we assume vector entries in both $x\in\R^n$ and $y\in\R^N$ are $0$-based indexed, e.g.\ $x = (x_0,\ldots,x_{n-1})^T$ and $y = (y_0,\ldots,y_{N-1})^T$. Conceptually we create virtual streams $J_0,\ldots,J_{{\log_b N}}$. For the stream $J_r$ we solve an instance of partition point query with $b^r$ sets in the partition defined as follows. At the level $r = \log_b N$, the partition is $\mathcal{P}_{\log_b N} = \mathcal{P} = \{S_0,\ldots,S_{N-1}\}$, with oracle $\oracle_{\log_b N} = \oracle$. For $0 \le r < \log_b N$, the partitions are defined inductively: we imagine a complete $b$-ary tree on $\{0,\ldots,N-1\}$ where the $j$th leaf node denotes $S_j$, and for $r < \log_b N$ each internal node represents a set equal to the union of the sets at its children. Said succinctly, for $r<\log_b N$, the oracle evaluation $\oracle_r(i)$ equals $\oracle(i) / b^{\log_b N - r}$. Note that since $b = 2^\ell$ is a power of $2$, $\oracle_r$ can be implemented in constant time by bitshifting right $\ell(\log_b N - r)$ positions ($\log_b N$ is a fixed value that can be computed once and stored during initialization).

When an update $(i,v)$ is seen in the actual stream, it is fed into a \pcs $P_r$ at level $r$ for each $0\le r\le \log_b N$, with partition $\mathcal{P}_r$ and oracle $\oracle_r$. Each $P_r$ is chosen to be a $(1/20, 3, \eps/60)$-partition point query structure, i.e.\ the error parameter is $\eps' = \eps/60$. Set $Q = 2b(\log_b N)/\eps''^2$, which is an upper bound on the total number of point queries we make to all $P_r$ combined. We set the failure probability for each $P_r$ to be $\eta = \delta/Q = \Theta(\eps^2\delta / (b\log_b N))$. 

To answer a heavy hitters query, we walk down the tree as follows (our goal is to have, at all times $0\le r\le \log_b N$, a list $L_r$ of size $O(1/\eps^2)$ which contains all the $\eps$-heavy hitter partitions at level $r$). First we set $L_0 = \{0\}$. Then for $r=1$ to $\log_b N$, we point query every child of an element of $L_{r-1}$. We then set $L_r$ to be the nodes at level $r$ whose point query result were amongst the $2/\eps''^2$ highest for $\eps'' = \eps'/3 = \eps/180$. We finally return $L = L_{\log_b N}$ as our set of partition heavy hitters.

\begin{theorem}\TheoremName{btree}
Suppose $1/N^c \le \eps < 1/2$ and $\delta\in(0,1/2)$. Then the $b$-tree as described above produces a correct output $L$ for partition $\eps$-heavy hitters with probability at least $1-\delta$ for any $2\le b\le N$. Furthermore, for any constant $c'\ge 1$ and any constant $0 < \gamma < \gamma_0$ where $\gamma_0$ depends on $c,c'$, for any failure probability $1/N^{c'} \le \delta < 1/2$, there is a choice of $b$ so that the $b$-tree uses space $O(\eps^{-2} \log N)$, has update time $O(\log N)$, and query time $O( (\eps^{-2}\log N)((\log N)/(\eps\delta))^{\gamma})$. The constants in the big-Oh depend on $1/\gamma$.
\end{theorem}
\begin{proof}
We have $|L| \le 2/\eps''^2 = O(1/\eps^2)$ by definition of the algorithm. Note that we make at most $Q$ point queries since each $L_r$ has size at most $2/\eps''^2$, and each element of $L_r$ has at most $b$ children. Also, there are $\log_b N$ values of $r$ for which we query children (since the leaves have no children). Thus the total number of queries is at most $2b(\log_b N)/\eps''^2 = Q$. Now we show that, conditioned on the event that all $Q$ point queries succeed, correctness holds by induction on $r$. This would complete the correctness analysis, since the probability that any point query fails at all is at most $Q\eta = \delta$, by a union bound. For the induction, our inductive hypothesis is that $L_r$ contains every $\eps$-partition heavy hitter at level $r$. This is true for $r = 0$ since $L_0 =\{0\}$, and the root of the tree contains only a single partition. Now, assume $L_{r-1}$ satisfies the inductive hypothesis. Since any ancestor of a partition $\eps$-heavy hitter is itself an $\eps$-partition heavy hitter, it follows that every partition heavy hitter at level $r$ is the child of some set in $L_{r-1}$. Thus if we let $A$ be the collection of children of sets in $L_{r-1}$, \Lemma{take-top} implies that $L_r$ also contains all the partition $\eps$-heavy hitters

As for space and time bounds, we choose $b$ to be $((\log N)/(\eps\delta))^\gamma$, rounded up to the nearest integer power of $2$. This is at most $N$ for $\gamma$ sufficiently small since $\eps,\delta > 1/\poly(N)$, and at least $2$ for $N$ larger than some constant depending on $\gamma$ (which we can assume without loss of generality, by padding the partition with empty sets). By \Lemma{pcs} and choice of $b$, the space is asymptotically
$$
\frac 1{\eps^2}\cdot \frac{\log N}{\log b}\log(1/\eta) \simeq \frac 1{\eps^2}\cdot \frac{\log N}{\log b} \cdot (\log(1/\delta) + \log b + \log\log N + \log(1/\eps)) \simeq \eps^{-2}\log N .
$$
The update time is asymptotically
$$
\frac{\log N}{\log b}\log(1/\eta) \simeq \frac{\log N}{\log b} \cdot (\log(1/\delta) + \log b + \log\log N + \log(1/\eps)) \simeq \log N .
$$
The query time is asymptotically
$$
\frac b{\eps^2}\cdot \frac{\log N}{\log b}\log(1/\eta) \simeq \frac b{\eps^2}\cdot \frac{\log N}{\log b} \cdot (\log(1/\delta) + \log b + \log\log N + \log(1/\eps)) \simeq \eps^{-2}\log N \cdot (\frac{\log N}{\eps\delta})^\gamma.
$$
\end{proof}

\begin{corollary}\CorollaryName{btree}
For any constant $0<\gamma<1/2$, (non-partition) heavy hitters with error parameter $\eps\in(0,1/2)$ and failure probability $1/\poly(n) < \delta < 1/2$ can be solved with space $O(\eps^{-2}\log n)$, update time $O(\log n)$, query time $O(\eps^{-2}\log n ((\log n)/(\eps\delta))^{\gamma})$, and failure probability $\delta$.
\end{corollary}
\begin{proof}
We can assume $\eps > 1/\sqrt{n}$ without loss of generality, since otherwise the trivial solution of keeping $x$ in memory explicitly suffices. Otherwise, the claim follows from \Theorem{btree} by considering the partition $\mathcal{P} = \{\{1\},\ldots,\{n\}\}$ with $N=n$ and $\oracle(i) = i$.
\end{proof}

\section{Connection between heavy hitters and list-recoverable codes}\SectionName{list-recovery}

Our approach is somewhat related to list-recoverable codes, which were first used in group testing in \cite{Cheraghchi09,IndykNR10} and in compressed sensing in \cite{NgoPR12} (and also in subsequent compressing works, e.g. \cite{GilbertNPRS13,GilbertLPS14}). We say a code $\mathcal{C}\subset [q]^m$ is {\em $(\epsilon, \ell, L)$-list recoverable} if for any sequence of lists $L_1,\ldots,L_m \subset[q]$ with $|L_j| \le \ell$ for all $j\in[m]$, there are at most $L$ codewords in $\mathcal{C}$ whose $j$th symbol appears in at least a $(1-\epsilon)$-fraction of the lists $L_j$. A decoding algorithm is then given these $m$ lists and must find these at most $L$ codewords. To see the connection to our current heavy hitters problem, it is known (via the probabilistic method) that such codes exist with $|\mathcal{C}| \ge n$, $q = \poly(\log n)$, $m = O(\log n/\log\log n)$, and $\ell, L = O(\log n)$, where $\epsilon$ can be made an arbitrarily small constant (see for example the first row of \cite[Figure 1]{HemenwayW15}). Suppose we had such a code $\mathcal{C}$ with encoding function $\enc:[n]\rightarrow \mathcal{C}$. Then our heavy hitters algorithm could follow the scheme of \Figure{basic}. That is, we would instantiate $m$ $b$-trees $P_1,\ldots,P_m$ for partition heavy hitters with $\oracle_j(i) = \enc(i)_j$. By picking constants appropriately in the parameter settings, one can ensure that whp at most an $\epsilon$-fraction of the $P_j$ fail. Thus, whp every heavy hitter appears in a $(1-\epsilon)$-fraction of the lists $L_j$, and we would then perform list-recovery decoding to find all the heavy hitters. The trouble with this approach is that there currently are no {\em explicit} codes achieving these parameter settings, let alone with linear time encoding and fast decoding. Indeed, this is the source of slight suboptimality in the $\ell_1/\ell_1$ ``for all'' compressed sensing scheme of \cite{GilbertLPS14}.

The key to our progress is to sidestep the lack of explicit optimal list-recoverable codes with linear-time encoding and fast decoding by realizing that list-recoverability is stronger than what we actually need. First of all, list-recovery says that {\em for all} choices of lists $L_1,\ldots,L_m$, some condition holds. In our case the lists are random (symbols in it contain concatenations with random hash evaluations $h_j$), since our heavy hitters algorithm is allowed to be a randomized algorithm. Secondly, our decoding algorithm does not need to handle arbitrary lists, but rather lists in which each symbol has a distinct name (recall from \Section{overview} that the ``name'' of $z(i)_j$ is $h_j(i)$). This is because, in our problem, with good probability any list $L_j$ has the property that for each heavy hitter $i$ with $z(i)_j \in L_j$, $z(i)_j$ will be reported by a partition point query data structure to have much heavier weight than any other $z\in L_j$ with the same name. Thus our post-processing step using $Q_j$ allows us to only have to perform decoding from lists with a certain structure which does not exist in the general list-recovery problem, namely that no two symbols in $L_j$ agree on the first few (namely $O(\log\log n)$) bits.

\input{strictturnstile.tex}

%% file: strictturnstile.tex
\section{A simpler query algorithm for strict turnstile $\ell_1$ heavy hitters}\SectionName{strict}
In this section, we describe our simpler algorithm for strict turnstile $\ell_1$ heavy hitters. The high-level idea is the following: Recall the general turnstile algorithm from \Section{turnstile}. There we end up having to find $\eps$-spectral clusters. This subproblem arose because several data structures $P_{j}^k$ may err and thereby insert spurious edges into the graph $G$. Now in the strict turnstile case, we can roughly ensure that we can recognize when a $P_j^k$ errs. By simply deleting all edges returned by such a $P_j^k$, the remaining graph $G$ has each heavy hitter corresponding to a number of connected components that are completely disjoint from the subgraphs corresponding to other heavy hitters and noise from light elements. Furthermore, for each heavy hitter, one of these connected components has at least 90\% of the corresponding codeword $\enc(i)$. We can thus replace the cluster finding algorithm by a simple connected components algorithm and thereby obtain a much simpler algorithm. This section gives all the details of the above.

Our solution needs the more standard definition of heavy hitters: On a query, we must return a set $L$ containing all indices $i$ where $x_i \geq \eps\|x\|_1$ and no indices $j$ such that $x_j < (\eps/2)\|x\|_1$. 

The first thing our solution in the strict turnstile case needs, is the fact that one can maintain $\|x\|_1$ exactly in $O(1)$ space and with $O(1)$ update time. This follows trivially by maintaining a counter $C$ that is initialized as $0$, and upon every update $(i, \Delta)$, we update $C \gets C+\Delta$. Since we are in the strict turnstile model, we will always have $C = \|x\|_1$. 

The second thing is a variant of partition point queries that have been adapted to the $\ell_1$ strict turnstile case.

\begin{definition}
In the $\ell_1$ \emph{strict turnstile partition point query problem} there is some error parameter $\eps \in (0,1/2)$. There is also some partition $\mathcal{P} = \{S_1,\ldots,S_N\}$ of $[n]$, and it is presented in the form of an oracle $\oracle:[n]\rightarrow[N]$ such that for any $i\in[n]$, $\oracle(i)$ gives the index $j\in[N]$ such that $i\in S_j$. Define a vector $y\in\R^N$ such that for each $j\in[N]$,
$$
y_j = \sum_{i\in S_j} x_i.
$$
On a query for index $j\in[N]$, we must output a value $\tilde{y}_j$ such that $y_j \le \tilde{y}_j \le y_j + \eps\|y_{\proj{1/\eps}}\|_1$.
\end{definition}

\begin{lemma}
\LemmaName{l1partquery}
For any failure probability $1/\poly(n) < \delta < 1/2$, there is a solution to the strict turnstile partition point query problem with space $O(\eps^{-1} \log(1/\delta))$ words, update time $O(\log (1/\delta))$ and query time $O(\log(1/\delta))$. Furthermore, the solution guarantees that even when it errs, the returned estimate $\tilde{y}_j$ is at least $y_j$. 
\end{lemma}

\begin{proof}
Simply implement the \cm sketch of \cite{CormodeM05} on the vector $y$. Thus on an update $(i,\Delta)$, feed the update $(\oracle(i),\Delta)$ to the \cm sketch and on a query for $y_j$, return the estimate for $j$ in the \cm sketch.
\end{proof}

Using \Lemma{l1partquery}, we can also modify the $b$-tree from \Section{btree} to obtain a no \emph{false negatives} guarantee. More specifically, define a new version of heavy hitters called \emph{threshold heavy hitters with no false negatives}. We say that an algorithm solves threshold $\ell_1$ heavy hitters with no false negatives and failure probability $\delta$, if it supports taking a query value $\phi>0$. If the query value $\phi$ is less than $\eps \|x_{\proj{1/\eps}}\|_1$, then the algorithm may return an arbitrary answer. If $\phi \geq \eps \|x_{\proj{1/\eps}}\|_1$, then with probability at least $1-\delta$, it must return a set $L$ that has size $O(\phi^{-1}\|x\|_1)$ and contains all indices $i$ with $x_i \geq \phi$ and no indices $j$ with $x_j \leq \phi-(\eps/2)\|x_{\proj{1/\eps}}\|_1$. With the remaining probability at most $\delta$, it returns the empty set.

Such algorithms can thus recognize when they err (indicated by an empty return set) on queries for sufficiently large $\phi$. The following lemma essentially follows from the $b$-tree solution in \Theorem{btree} combined with \Lemma{l1partquery}:
\begin{lemma}
\LemmaName{thresholdnoneg}
Suppose $1/n^c \le \delta < 1/2$ where $c>0$. Let $\gamma$ be any constant satisfying $0 < \gamma < \gamma_0$ for some fixed constant $\gamma_0$. Then there is a modification of the $b$-tree described in \Section{btree} and a choice of $b$, which in the strict turnstile $\ell_1$ setting solves threshold heavy hitters with no false negatives and failure probability at most $\delta$. Furthermore, the $b$-tree uses space $O(\eps^{-1} \log n)$, has update time $O(\log n)$, and query time $O( (\eps^{-1}\log n)((\log n)/(\eps\delta))^{\gamma})$. The constants in the big-Oh depend on $1/\gamma$. 
\end{lemma}

\begin{proof}
Take the $b$-tree solution as described in \Section{btree}. Replace each \pcs for $P_r$ with the algorithm from \Lemma{l1partquery} with approximation factor $\eps'=\eps/2$. This means that at each level of the $b$-tree, we can do a point query for the total mass of a subtree. By \Lemma{l1partquery} these estimates are never underestimates, even when the algorithm errs. And when the algorithm does not err, the returned estimate is at most $\eps'\|y_{\proj{1/\eps}}\|_1 \leq \eps'\|x_{\proj{1/\eps}}\|_1 = (\eps/2)\|x_{\proj{1/\eps}}\|_1$ too high. We set the failure probability for each data structure from \Lemma{l1partquery} to $O(\eps \delta /( b \lg_b n))$.

On a query for all indices $i$ with $x_i \geq \phi$, we traverse the tree roughly as described in \Section{btree}. The main difference is that when choosing the set $L_r$ after having point queried every child of nodes in $L_{r-1}$, we choose $L_r$ to be all nodes where the point query returned an estimate of at least $\phi$. In case $|L_r| > 3 \eps^{-1}$ at any level, we abort and return $L = \emptyset$.

Now observe that if no query to the children of nodes in $L_{r-1}$ err, then for $\phi  \geq \eps \|x_{\proj{1/\eps}}\|_1$, there can be no more than $3\eps^{-1}$ nodes for which the estimate is at least $\phi$. This follows since if the point queries did not err, then the true mass of those subtrees must be at least $\phi-(\eps/2) \|x_{\proj{1/\eps}}\|_1 \geq (\eps/2) \|x_{\proj{1/\eps}}\|_1$. There can be at most $3 \eps^{-1}$ such subtrees. Thus conditioned on not making an error in any point query on the first $r-1$ levels of the $b$-tree, we will visit at most $3 \eps^{-1}$ nodes on level $r$. It follows that we only abort and return an empty list in case some query errs. And if no query errs, we will ask only $O(\eps^{-1} b \lg_b n)$ queries. Since we set the failure probability in each algorithm from \Lemma{l1partquery} to $O(\eps \delta /( b \lg_b n))$, the lemma follows.
\end{proof}

We only need one more ingredient in our solution, namely a property of edge-expander graphs. We have not been able to find a reference for this result, but suspect it is known.

\begin{lemma}\LemmaName{big-component}
If $G  = (V,E)$ is a $d$-regular $\delta$-edge-expander, then for any constant $0<\mu<1$, if we remove any set of at most $(\delta \mu /2)|V|$ vertices in $G$ and their incident edges, then there remains a connected component of size at least $(1-\mu)|V|$.
\end{lemma}
\begin{proof}
Let $S \subseteq V$ be any set of at most $(\delta \mu /2)|V|$ vertices in $G$. The number of edges with an endpoint in $S$ is at most $d|S|$ since $G$ is $d$-regular. Let $C_1,\dots,C_m \subseteq V$ be the connected components in $G_{V \setminus S}$, where $G_{V \setminus S}$ is the graph obtained from $G$ by removing the vertices $S$ and all edges incident to $S$. Assume for contradiction that $|C_i| < (1-\mu)|V|$ for all $C_i$. For each $C_i$, let $C_i^*$ denote the smaller of $C_i$ and $V \setminus C_i$. It follows from $|C_i|$ being less than $(1-\mu)|V|$ that 
$$\frac{\mu}{1-\mu}\cdot |C_i| < |C_i^*| \leq \min\left\{\frac{|V|}2,|C_i|\right\} .$$
Since $G$ is a $\delta$-edge-expander and $|C_i^*| \leq |V|/2$, there are at least 
$$\delta d |C_i^*| \geq \delta \min\left\{1,\frac{\mu}{1-\mu}\right\}\cdot d|C_i| \geq \delta \min\{1,\mu\}d|C_i| = \delta \mu d |C_i|$$
edges across the cut $(C_i^*, V \setminus C_i^*)$ in $G$. This cut is just $(C_i, V \setminus C_i)$. But $C_i$ is a connected component in $G_{V \setminus S}$ and thus all edges leaving $C_i$ in $G$ must go to a node in $S$. Therefore, the number of edges incident to a node in $S$ is at least $\sum_i \delta \mu d|C_i| = \delta \mu d(|V|-|S|)$. But $|S| \leq (\delta \mu/2)|V|$ and hence 
$$|V|-|S| \geq (1-\delta \mu/2)|V| \geq \frac{1-\delta \mu/2}{\delta \mu /2}\cdot |S| .$$
We thus have at least 
$$\delta \mu \frac{1-\delta \mu/2}{\delta \mu /2}\cdot d|S|= 2 (1-\delta \mu/2) d |S| > d|S|$$
edges incident to $S$. But this contradicts that there are only $d|S|$ edges incident to $S$.
\end{proof}

The following result due to Alon et al.~\cite{alon:expanders} shows that there exists efficient explicit edge-expanders:
\begin{theorem}[Alon et al.~\cite{alon:expanders}]
\TheoremName{edgeexpander}
There exists a constant $\delta>0$, such that for any positive integer $n$, there exists a $12$-regular $\delta$-edge-expander $G=(V,E)$ with $n/2 \leq |V| \leq n$. Furthermore, $G$ is constructable in time $O(\poly(n))$.
\end{theorem}

We have set the stage for the description of our simpler $\ell_1$ strict turnstile solution. Our new algorithm is described in the following:

\paragraph{Strict Turnstile Algorithm.}
If $\eps = o(1/\lg n)$, we re-execute the reduction in \Section{reduction}, modified to the $\ell_1$ setting. More specifically, we hash to $z = \Theta(1/(\eps \lg(1/n)))$ subproblems. Each subproblem is solved with a threshold heavy hitters structure with $\eps = O(1/\lg n)$ and failure probability $1/\poly(n)$. Thus on a query to report all heavy hitters, we use the fact that we know $\|x\|_1$ exactly in order to query each threshold heavy hitters structure with $\phi = \eps \|x\|_1$. By an analysis similar to \Section{reduction}, with probability at least $1-1/\poly(n)$, this will ensure $\phi \geq \eps\|y_{\proj{1/\eps}}\|_1$ where $y$ denotes the vector corresponding to any of these subproblem. If $\eps$ is $\Omega(1/\lg n)$, we can also reduce to threshold heavy hitters with the same $\eps$ and failure probability $1/\poly(n)$ simply by maintaining one structure and querying it with $\phi=\eps\|x\|_1$. Thus we have now reduced the problem to solving threshold heavy hitters with $\eps = \Omega(1/\lg n)$ in the $\ell_1$ strict turnstile case, where we are guaranteed that the query value $\phi$ is at least $\eps\|x_{\proj{1/\eps}}\|_1$.

We now focus on solving threshold heavy hitters with $\eps = \Omega(1/ \lg n)$ and failure probability $1/\poly(n)$ for a query value $\phi \geq \eps\|x_{\proj{1/\eps}}\|_1$. Following \Section{turnstile}, we construct the 12-regular $\delta$-edge-expander from \Theorem{edgeexpander} with $m= \Theta(\lg n/\lg \lg n)$ vertices. Call it $F$. This takes $\poly \lg n$ time. We store it in adjacency list representation, consuming $O(\lg n/\lg \lg n)$ words of memory.

We then pick hash functions $h_1,\dots,h_m : [n] \to [\poly(\lg n)]$ independently from a pairwise independent family and instantiate data structures $D_1,\dots,D_m$. Each $D_i$ is the threshold heavy hitters with no false negatives from \Lemma{thresholdnoneg}. The failure probability of each structure is $1/\poly(\lg n)$ and the approximation factor is $\eps$. 

On an update $(i,\Delta)$, we compute in $O(\lg n)$ time an encoding $\enc(i)$ of $i$ into a string of length $T = O(\lg n)$ bits by an error-correcting code with constant rate that can correct an $\Omega(1)$-fraction of errors (as in \Section{turnstile}). We again partition $\enc(i)$ into $m$ contiguous bitstrings each of length $t = T/m = \Theta(\lg \lg n)$. We let $\enc(i)_j$ denote the $j$th bitstring. For each $j=1,\dots,m$, we feed the update $(h_j(i) \circ \enc(i)_j \circ h_{\Gamma(j)_1} \circ \cdots \circ h_{\Gamma(j)_{12}}, \Delta)$ to $D_j$ (similarly to \Section{turnstile}).

On a query with threshold $\phi >0$, we query each $D_j$ with $\phi$. Let $L_j$ be the output of $D_j$. If $L_j$ is non-empty, we go through all returned values and check if there are two with the same name (same $h_j(i)$). If so, we overwrite $L_j \gets \emptyset$ and say that $D_j$ failed.  Now in all remaining $L_j$, all names are unique. From these lists, we construct the same layered graph $G$ as described in \Section{overview}. Recall that we only add edges if both endpoints suggest it. Now the crucial observation which makes the $\ell_1$ strict turnstile easier to solve, is that the $D_j$'s guarantee no false negatives. Thus if some $D_j$ did not err and did not fail, then no node in $G$ corresponding to a name $h_j(i)$ returned by $D_j$, can be connected by an edge to a node $h_{j'}(i')$ returned by a $D_{j'}$ for any $i'\neq i$. To see this, notice that the actual index $h_{j'}(i) \circ \enc(i)_{j'} \circ h_{\Gamma(j')_1} \circ \cdots \circ h_{\Gamma(j')_{12}}$ corresponding to $h_{j'}(i)$ must not be in $L_{j'}$ for this to happen. But we are guaranteed no false negatives, so if $D_{j'}$ returns anything, it must have $h_{j'}(i) \circ \enc(i)_{j'} \circ h_{\Gamma(j')_1} \circ \cdots \circ h_{\Gamma(j')_{12}}$ in its output. To summarize, for every index $i$ with $x_i > \phi$, all $D_j$ that did not err and did not fail produce nodes corresponding to names $h_j(i)$, and these nodes are not connected to nodes of name $h_{j'}(i')$ for an $i'\neq i$. Furthermore, for two $D_j$ and $D_{j'}$ that both did not err and fail, the edges between the returned nodes are in the graph $G$. But the hash function $h_j$ have collisions with probability $1/\poly(\log n)$ and the data structures $D_j$ have failure probability $1/\poly(\lg n)$. Thus with probability $1-1/\poly(n)$, there are at most a $\gamma$-fraction of the $D_j$ that either fail or err. Here we can choose $\gamma$ as an arbitrarily small constant by increasing the polylogs in $1/\poly(\log n)$ failure probability. By \Lemma{big-component}, we get that with probability $1-1/\poly(n)$, each index $i$ with $x_i \geq \phi$ has a connected component in $G$ of size at least $0.9m$. We now finish off by computing all connected components in $G$. For each component of size at least $0.9 m$, we stitch together the pieces $\enc(i)_j$ to obtain 90\% of the codeword $\enc(i)$. We decode $i$ in linear time in its bit length, which is $O(\lg n)$ time. Finally, we verify all these returned indices by doing a point query against a separate \cm structure with failure probability $1/\poly(n)$ and the same $\eps$. This also costs $O(\lg n)$ time per index.

Since we are guaranteed that each returned list $L_j$ has size at most $\phi^{-1}\|x\|_1$, the query time is $O(\phi^{-1}\|x\|_1 \lg n)$ for decoding codewords and for verifying the recovered indices (the norm of the vector represented by each $D_j$ is precisely the same as the norm of $x$). Now $\phi$ is $\eps \|x\|_1$, so the time spent is $O(\eps^{-1} \lg n)$. For querying the structures $D_j$, \Lemma{thresholdnoneg} gives us a query time of $O( (\eps^{-1}\log \log n)((\log \log n \poly(\log n))/\eps))^{\gamma})$ per $D_j$. Since $\eps = \Omega(1/\log n)$, we can choose $\gamma$ sufficiently small to obtain a query time of $\eps^{-1}\log^{\gamma'} n$ for any constant $\gamma'>0$. Since we have $m=O(\lg n/\lg \lg n)$, the total time for querying all $D_j$ is $O(\eps^{-1}\lg^{1+\gamma}n)$ for any constant $\gamma>0$.

For the update time, observe that each $D_j$ has update time $O(\lg \lg n)$ and there are $O(\lg n/\lg \lg n)$ such structures to be updated, for a total of $O(\lg n)$ time. Computing $\enc(i)$ also takes $O(\lg n)$ time. The space is optimal $O(\eps^{-1} \lg n)$ words.

\begin{theorem}\TheoremName{l1strictresult}
For any constant $0<\gamma<1/2$, heavy hitters with error parameter $\eps\in(0,1/2)$ and failure probability $1/\poly(n) < \delta < 1/2$ can be solved in the strict $\ell_1$ turnstile setting with optimal $O(\eps^{-1}\log n)$ space, update time $O(\log n)$, query time $O(\eps^{-1}\log^{1+\gamma} n)$, and failure probability $\delta$.
\end{theorem}

\section{Expected Time for Strict Turnstile $\ell_1$}\SectionName{expected-time}
If we are satisfied with expected query time, then there is an even simpler and very practical algorithm for the strict turnstile $\ell_1$ case. The algorithm is based on the $b$-tree, but with $b=2$, i.e. a binary tree. As in \Section{strict}, we solve the more standard variant of heavy hitters in which we must return a set $L$ containing all indices $i$ where $x_i \geq \eps\|x\|_1$ and no indices $j$ such that $x_j < (\eps/2)\|x\|_1$. 

For what follows, we assume vector entries in $x\in\R^n$ are $0$-based indexed, e.g.\ $x = (x_0,\ldots,x_{n-1})^T$. Conceptually we create virtual streams $J_0,\ldots,J_{{\log_2 n}}$. For the stream $J_r$ we solve an instance of partition point query with $2^r$ sets in the partition defined as follows. At the level $r = \log_2 n$, the partition is $\mathcal{P}_{\log_2 n} = \mathcal{P} = \{S_0,\ldots,S_{n-1}\}$, with oracle $\oracle_{\log_2 n} = \oracle$. For $0 \le r < \log_2 n$, the partitions are defined inductively: we imagine a complete binary tree on $\{0,\ldots,n-1\}$ where the $j$th leaf node denotes $S_j$, and for $r < \log_2 n$ each internal node represents a set equal to the union of the sets at its children. Said succinctly, for $r<\log_2 n$, the oracle evaluation $\oracle_r(i)$ equals $\oracle(i) / 2^{\log_2 n - r}$. Note that $\oracle_r$ can be implemented in constant time by bitshifting right $\ell(\log_2 n - r)$ positions ($\log_2 n$ is a fixed value that can be computed once and stored during initialization).

When an update $(i,v)$ is seen in the actual stream, it is fed into a \pcs $P_r$ at level $r$ for each $0\le r\le \log_2 n$, with partition $\mathcal{P}_r$ and oracle $\oracle_r$. Each $P_r$ is chosen to be the partition point query structure from \Lemma{l1partquery} with error parameter $\eps' = \eps/2$ and failure probability $\delta = 1/4$. This has each $P_r$ using $O(\eps^{-1})$ words of space, supporting updates and queries in $O(1)$ time. We also maintain one \cm sketch on $x$ with error parameter $\eps' = \eps/2$ and failure probability $1/\poly(n)$. 

The space is $O(\eps^{-1}\lg n)$ and the update time is $O(\lg n)$. 

To answer a heavy hitters query, we walk down the tree as follows (our goal is to have, at all times $0\le r\le \log_2 n$, a list $L_r$ of size $O(1/\eps^{-1})$ which contains all the $\eps$-heavy hitter partitions at level $r$). First we set $L_0 = \{0\}$. Then for $r=1$ to $\log_2 n$, we point query every child of an element of $L_{r-1}$. We then set $L_r$ to be the nodes at level $r$ whose point query result was at least $\eps\|x\|_1$. Recall from \Section{strict} that in the strict turnstile $\ell_1$ setting, we can maintain $\|x\|_1$ exactly. Once we have computed $L_{\log_2 n}$, we point query each element in the list against the \cm sketch and filter out those indices where the returned value is less than $\eps \|x\|_1$.

The key insight is that the data structure from \Lemma{l1partquery} never returns an underestimate. Thus the list $L_{\log_2 n}$ will always contain all indices $i$ with $x_i \geq \eps\|x\|_1$. The only issue is that we may consider more candidate heavy hitters than really exist, which might increase the query time $t_q$. We show that in fact $t_q$ remains small {\em in expectation}.

Let the $2n-1$ nodes in the conceptual binary tree be called $u_1,\ldots,u_{2n-1}$. During query, we traverse this tree starting at the root. While visiting a node $u$ at level $r$, we query it in $P_r$ and only recurse to its children if the subtree corresponding to $u$ is declared to have mass at least $\eps\|x\|_1$ at level $r$. For $j\in[2n-1]$ let $Y_j$ be an indicator random variable for whether $P_r$ declares $u_j$ as having mass at least $\eps\|x\|_1$. Note the runtime of a query is $\sum_j Y_j$. Thus the expected runtime satisfies
$$
t_q = \E \sum_{j=1}^{2n-1} Y_j  = \sum_j \E Y_j .
$$

Now, what is $\E Y_j$? Call a node $u_j$ which {\em actually} corresponds to a subtree of mass at least $(\eps/2)\|x\|_1$ \emph{heavy}. Then for heavy $u_j$, we use the simple upper bound $\E[Y_j] \leq 1$; there are at most $2\eps^{-1}\lg n$ such heavy $j$ summed over the entire tree. Otherwise, if $d$ is the distance to the lowest ancestor $v$ of $u_j$ which is heavy, then for $Y_j$ to be $1$, it must be that every point query to an ancestor of $u_j$ below $v$ must fail. This happens with probability $1/4^{d-1}$ since the data structures are independent across the levels of the tree. However, if we fix $d$ and ask how many nodes in the entire tree have shortest distance to a heavy ancestor equal to $d$, then that is at most $(2\eps^{-1}\lg n )\cdot 2^{d-1}$.

Therefore the expected query time $t_q$ for traversing the tree is:
$$
\sum_j \E Y_j \le (\eps^{-1}\lg n) \cdot \sum_{r=0}^\infty \left(\frac 24\right)^r = O(\eps^{-1}\lg n) .
$$
Finally, we also have $|L_{\lg_2 n}|$ queries to a \cm sketch with failure probability $1/\poly(n)$. These queries cost $O(\lg n)$ time each. Thus we need to bound $O(\lg n) \cdot \E |L_{\lg_2 n}|$. By an argument similar to the above, consider a leaf $u_j$ corresponding to an index $j$ with $x_j \leq (\eps/2)\|x\|_1$. Let $d$ denote the distance from $u_j$ to its nearest heavy ancestor. For $u_j$ to be in $L_{\lg_2 n}$, every point query on this path must fail. This happens with probability $1/4^{d-1}$. Now how many \emph{leaves} can have shortest distance to a heavy ancestor equal to $d$? Since we are only considering leaves, ancestors of distance $d$ are at the exact same level for all leaves. But a level has at most $2 \eps^{-1}$ heavy nodes, and thus there are at most $(2 \eps^{-1}) \cdot 2^{d-1}$ leaves with distance exactly $d$ to their nearest heavy ancestor. By an argument similar to the above, we get that $\E |L_{\lg_2 n}| = O(\eps^{-1})$. We conclude:
\begin{theorem}\TheoremName{l1strictexpected}
For any constant $0<\gamma<1/2$, heavy hitters with error parameter $\eps\in(0,1/2)$ and failure probability $1/\poly(n) < \delta < 1/2$ can be solved in the strict $\ell_1$ turnstile setting with optimal $O(\eps^{-1}\log n)$ space, update time $O(\log n)$, expected query time $O(\eps^{-1}\log n)$, and failure probability $1/\poly(n)$.
\end{theorem}